\documentclass[11pt,letterpaper]{article}
\usepackage[margin=1in]{geometry}
\usepackage{amssymb,amsmath,amsthm}
\usepackage{graphicx}
\usepackage[colorlinks,citecolor=blue,linkcolor=blue,urlcolor=red,pagebackref]{hyperref}
\usepackage{subcaption}
\usepackage{booktabs}
\usepackage{lineno}
\usepackage{color}
\usepackage{enumitem}
\usepackage{tikz}
\usepackage{algorithm}
\usepackage[noend]{algpseudocode}
\usepackage{framed}
\usepackage{thm-restate}
\usepackage[T1]{fontenc}
\usepackage[colorinlistoftodos]{todonotes}

\newcommand*\patchAmsMathEnvironmentForLineno[1]{%
  \expandafter\let\csname old#1\expandafter\endcsname\csname #1\endcsname
  \expandafter\let\csname oldend#1\expandafter\endcsname\csname end#1\endcsname
  \renewenvironment{#1}%
     {\linenomath\csname old#1\endcsname}%
     {\csname oldend#1\endcsname\endlinenomath}}%
\newcommand*\patchBothAmsMathEnvironmentsForLineno[1]{%
  \patchAmsMathEnvironmentForLineno{#1}%
  \patchAmsMathEnvironmentForLineno{#1*}}%
\AtBeginDocument{%
\patchBothAmsMathEnvironmentsForLineno{equation}%
\patchBothAmsMathEnvironmentsForLineno{align}%
\patchBothAmsMathEnvironmentsForLineno{flalign}%
\patchBothAmsMathEnvironmentsForLineno{alignat}%
\patchBothAmsMathEnvironmentsForLineno{gather}%
\patchBothAmsMathEnvironmentsForLineno{multline}%
}

\newtheorem{theorem}{Theorem}[section]
\newtheorem{definition}[theorem]{Definition}

\newtheorem{lemma}[theorem]{\bfseries{Lemma}}

\newtheorem{observation}[theorem]{\bfseries{Observation}}

\renewcommand{\mod}{\mathrm{\ mod\ }}

\newcommand{\eps}{\varepsilon}

\renewcommand{\leq}{\leqslant}
\renewcommand{\geq}{\geqslant}

\newbox\ProofSym
\setbox\ProofSym=\hbox{%
\unitlength=0.18ex%
\begin{picture}(10,10)
\put(0,0){\framebox(9,9){}}
\put(0,3){\framebox(6,6){}}
\end{picture}}

\begin{document}

\title{
    Long Arithmetic Progressions in Sumsets and Subset Sums: Constructive Proofs and Efficient Witnesses
}
\author{Lin Chen$^{a}$, Yuchen Mao$^{a}$, Guochuan Zhang$^{a}$ \\
        \small $^{a}$College of Computer Science, Zhejiang University, Hangzhou, China\\ 
        \tt{\{chenlin198662,maoyc,zgc\}@zju.edu.cn}  \\
       \\
}

\date{}

\maketitle

\begin{abstract}
    Existence of long arithmetic progression in sumsets and subset sums has been studied extensively in the field of additive combinatorics. These additive combinatorics results play a central role in the recent progress of fundamental problems in theoretical computer science including Knapsack and Subset Sum. The non-constructiveness of relevant additive combinatorics results affects their application in algorithms. In particular, additive combinatorics-based algorithms for Subset Sum, including an $\tilde{O}(n)$-time algorithm for dense subset sum [Bringmann and Wellnitz '21] and an $\tilde{O}(n + \sqrt{a_{\max}t})$-time algorithm [Chen, Lian, Mao, and Zhang '24], work only for the decision version of the problem, but not for the search version. To find a solution, one has to spend a lot more time.

    We provide constructive proofs for finite addition theorems [S{\'a}rk{\"o}zy'89 '94], which are fundamental results in additive combinatorics concerning the existence of long arithmetic progression in sumsets and subset sums. Our constructive proofs yield a near-linear time algorithm that returns an arithmetic progression explicitly, and moreover, for each term in the arithmetic progression, it also returns its representation as the sum of elements in the base set.

    As an application, we can obtain an $\tilde{O}(n)$-time algorithm for the search version of dense subset sum now. Another application of our result is Unbounded Subset Sum, where each input integer can be used an infinite number of times.  A classic result on the Frobenius problem [Erd\H{o}s and Graham '72] implies that for all $t \geq 2a^2_{\max}/n$, 
    the decision version can be solved trivially in linear time. It remains unknown whether the search version can be solved in the same time. Our result implies that for all $t \geq ca^2_{\max}/n$ for some constant $c$, a solution for Unbounded Subset Sum can be obtained in $O(n \log a_{\max})$ time.
 
     The major technical challenge is that the original proofs for the above-mentioned additive combinatorics results heavily rely on two fundamental theorems, Mann's theorem and Kneser's theorem. These two theorems constitute the main non-constructive part. To bypass these two obstacles, we introduce two techniques.
    \begin{itemize}
        \item A new set operation called greedy sumset. Greedy sumset computes a moderately large subset of the traditional sumset, but enjoys the advantage that searching for a representation for elements in the greedy sumset can be done efficiently.

        \item A framework that can be used to iteratively augment an arithmetic progression. It plays the role of Kneser's theorem in the proof but enjoys the advantage that the representation of elements in the arithmetic progression can be efficiently traced.
    \end{itemize}
\end{abstract}

\section{Introduction}
Let $A$ and $B$ be two sets of integers. Their sumset is
\[
    A + B = \{\textrm{$a + b$ : $a \in A$ and $b \in B$}\}.
\]
The $k$-fold sumset of $A$ is defined to be 
\[
    kA = \{\textrm{$a_1 + a_2 + \cdots + a_k$ : $a_i \in A$ for $1 \leq i\leq k$}\}.
\]
The set of subset sums of $A$ is 
\[
    \mathcal{S}(A) = \left\{\sum_{a \in A'}a : A' \subseteq A\right\}
\]
Let $P = \{s, s + d, s + 2d, \ldots, s + \ell d\}$. We say that $P$ is an arithmetic progression with common difference $d$ and length $\ell$. For the sake of clarity, we sometimes write $P$ as $\{s\} + \{0, d, 2d, \ldots, \ell d\}$.

The existence of long arithmetic progressions in $k$-fold sumsets $kA$ and subset sums $\mathcal{S}(A)$ has been studied extensively in additive combinatorics. There has been a long line of research in the literature~\cite{Fre93,Sar89,Sar94,Lev03,Lev97,SV05,SV06a,tao2008john,CFP21}. Interestingly, it also plays a central role in the recent progress of several fundamental problems in theoretical computer science.
\begin{itemize}
    \item Subset Sum. Given a (multi-)set $A$ of $n$ integers $\{a_1, \ldots, a_n\}$ and a target $t$, Subset Sum asks whether some subset of $A$ sums to $t$. Bringmann and Wellnitz~\cite{BW21} gave an $\tilde{O}(n)$-time\footnote{$\widetilde{O}(f)$ hides polylogarithmic factors in $f$.} algorithm for the case where the maximum input integer $\max_ia_i$ (denoted as $a_{\max}$), $\sum_ia_i$, $t$ and $n$ satisfy certain condition, which they call Dense Subset Sum. Chen, Lian, Mao and Zhang~\cite{chen2024improved} presented an $\widetilde{O}(n+\sqrt{a_{\max}t})$-time algorithm for the general case. They~\cite{chen2024approximating} also presented a weak approximation scheme with $\tilde{O}(n + \frac{1}{\eps})$ running time. Roughly speaking, the core argument underlying these results states that, if there is a long arithmetic progression in $\mathcal{S}(A)$ and $t$ happens to fall into this arithmetic progression, then the answer to Subset Sum is "yes".

    \item Knapsack. Given a knapsack of capacity $t$ and $n$ items $\{1, \ldots, n\}$, where item $i$ has a weight $w_i$ and a profit $p_i$, one should select items so as to maximize the total profit subject to the capacity constraint. Bringmann~\cite{Bri23a} and Jin~\cite{Jin23} independently obtained an $\tilde{O}(n + w_{\max}^2)$-time exact algorithm. 
    These results use the existence of the long arithmetic progressions in subset sums to obtain sharper proximity results, and utilize this proximity to design faster algorithms for Knapsack. 
\end{itemize}

So far, all the proofs of the existence of long arithmetic progressions in sumsets and subset sums are not constructive in the sense that they do not show how each term of the arithmetic progression is represented as a sum of integers from $A$. Therefore, to obtain such a representation, one has to compute $kA$ or $\mathcal{S}(A)$ explicitly, which incurs an unacceptable running time.

In certain applications like Knapsack, the existence of a long arithmetic progression suffices, because it is used to derive structural properties of the optimal solutions and there is no need to know how each term in the arithmetic progression is represented. In other applications like Subset Sum, however, without knowing how the terms of the arithmetic progression are represented, one cannot find a solution (that is, a subset that sums to $t$). As a consequence, both the $\tilde{O}(n)$-time algorithm for Dense Subset Sum and the $\tilde{O}(n + \sqrt{a_{\max}t})$-time algorithm for Subset Sum only solve the decision version of the corresponding problems. Indeed, both papers~\cite{BW21,chen2024approximating} raised the open question asking whether finding a solution can be done within the same running time.

As we have seen, unlike Knapsack, the current best-known algorithms for the decision version and search version of Subset Sum have a gap in time complexity. Such a phenomenon also exists in Unbounded Subset Sum. Let $a_1, \ldots, a_n$ be $n$ positive integers. Without loss of generality, assume that $a_1 < a_2 < \cdots < a_n$ and that their greatest common divisor is $1$. Given an integer $t$, the unbounded subset sum problem asks whether $t$ can be represented as $t = a_1x_1 + \cdots + a_nx_n$ for some non-negative integers $x_1, \ldots, x_n$. It is known that there exists some integer $F(a_1,\ldots,a_n)$, called the Frobenius number, such that the answer to Unbounded Subset Sum is always "yes" for any $t>F(a_1,\cdots,a_n)$.\footnote{More precisely, Frobenius number is defined as the largest integer $t$ for which the answer to Unbounded Subset Sum is "no".~\cite{brauer1942problem}} There is a rich literature dedicated to the Frobenius number, see, e.g.~\cite{heap1964graph,heap1965linear,hujter1987exact,kannan1992lattice,alfonsin2005diophantine}. Erd\H{o}s and Graham~\cite{erdos1972linear} proved that
\begin{equation}\label{eq:Frobenius}
    F(a_1,\cdots,a_n)\le 2a_{n-1}\lfloor\frac{a_n}{n}\rfloor-a_n=O(a_n^2/n)
\end{equation}
Their bound is currently the best known and is only a constant factor away from the actual value~\cite{aggarwal2024polynomial}. Therefore, for any $t \geq \frac{2a_na_{n-1}}{n}$, the answer to Unbounded Subset Sum is always "yes". 

Note that Unbounded Subset Sum is NP-hard in general, so Erd\H{o}s and Graham's bound characterizes a non-trivial region where Unbounded Subset Sum becomes easy. However, since their proof is non-constructive, it remains non-trivial to search for a solution within this "easy region".  Recently, Aggarwal, Joux, Santha and W{\k{e}}grzycki~\cite{aggarwal2024polynomial} showed that, one can find a solution for any $t\ge \eps a_n^2$ in time $\text{poly}(n,\log t)(\log (1/\eps))^{O(1/\eps)}$. Note that their algorithm only works for $t$ that is $\Omega(n)$ times larger than the Erd\H{o}s-Graham bound, and that the running time is also large (we remark that the main goal of~\cite{aggarwal2024polynomial} is not Unbounded Subset Sum itself, but rather its generalization for high-dimensional lattices). 

In general, knowing the existence of a solution does not necessarily make the
search problem easy to solve. There is a broad class of problems where a solution is guaranteed to exist, but searching for it becomes nontrivial~\cite{megiddo1991total}. The major question is: do Subset Sum and Unbounded Subset Sum belong to such problems? The issue that the decision version and the search version for Subset Sum and Unbounded Subset Sum have a gap in complexity seems to be mainly due to the application of additive combinatorics results. In particular, it brings breakthroughs for the decision problem, but not for the search problem due to the non-constructiveness.

Can we establish an "efficiently constructive" version of relevant additive combinatorics results so that not only do they apply to the decision problem, but also apply to the search problem? More precisely, we want to find a long arithmetic progression in sumsets and subset sums claimed in prior additive combinatorics results, and for each term in the arithmetic progression, we also want to know its representation as the sum of integers in $A$. Moreover, such information shall be computed in  time near-linear in the input size, because when these additive combinatorics results are applied to the decision problems, testing whether its condition is satisfied can be done trivially in linear time. Consequently, when the constructive version of these additive combinatorics results is applied to the search problem, we also want a similar linear time. 

This paper is dedicated to the target above. Before we proceed to describe our results, we first introduce some notions.




\subsection{Solutions and Witnesses}
Let $A$ and $B$ be two sets of integers. For an integer $z$, a solution for $z \in A + B$ is a pair $(a,b) \in A \times B$ such that $z = a+ b$.

For an integer $z$, a solution for $z \in kA$ consists of $k$ integers $(a_1, \ldots, a_k)$ such that
\[
    a_1 + \cdots + a_k = z \quad \text{and} \quad a_i \in A \,\,\text{for} \,\, 1 \leq i \leq k.
\]
For a set $S$ of integers, a witness for $S \subseteq kA$ is a data structure that can answer the following query: for any $s \in S$, return a solution for $s \in kA$. 

Similarly, for an integer $z$, a solution for $z \in \mathcal{S}(A)$ is a subset $A' \subseteq A$ such that $\sum_{a \in A'}a = z$. For a set $S$ of integers, a witness for $S \subseteq \mathcal{S}(A)$ is a data structure that can answer the following query: for any $s \in S$, return a solution for $s \in \mathcal{S}(A)$.

The query time of a witness is defined to be the time cost for answering a query.

\subsection{Our Contribution}
\subsubsection{Finite Addition Theorem I - Long Arithmetic Progressions in Sumsets}
Let $A \subseteq \mathbb{Z}[0, m]$ be a set of $n$ integers, where $\mathbb{Z}[0, m] = \{0, 1, \ldots, m\}$. In his seminal paper, S{\'a}rk{\"o}zy~\cite{Sar89} proved the finite addition theorem, which states that if $|A| \geq c\frac{m}{k}$ for some constant $c$, then the $k$-fold sumset $kA$ contains an arithmetic progression of length $m$. Eight years later, Lev~\cite{Lev97} improved the constant $c$. Below we present Lev's version.
\begin{theorem}[Finite Addition Theorem I, {\cite[Theorem 1]{Lev97}}]\label{thm:ka-old}
    Let $A \subseteq \mathbb{Z}[0, m]$ be a set of $n$ integers. Assume that $0 \in A$ and $\mathrm{gcd}(A) = 1$. Let $r$ be a positive integer.  If 
    \[
        n \geq \frac{m-1}{k} + 2,
    \]
    then there exists an arithmetic progression
    \begin{equation}\label{eq:ka-old}
        \{s\} + \{0, 1, \ldots, m\} \subseteq (2k - 1)A.
    \end{equation}
\end{theorem}
The assumption that $0 \in A$ and $\gcd(A) = 1$ is without loss of generality. One can always make $0 \in A$ by shifting the elements of $A$, and then make $\gcd(A) = 1$ by dividing the elements by their greatest common divisor. The original proof of Theorem~\ref{thm:ka-old} is non-constructive due to the non-constructiveness of its two major technical ingredients -- Mann's theorem~\cite{man42} and Kneser's theorem~\cite{kneser1953abschatzung} (see Subsection~\ref{subsec:tec-review}). Being two cornerstones in additive combinatorics, Mann's theorem and Kneser's theorem have led to many fundamental results~\cite{freiman1999structure,Sar89,Sar94,Lev03,Lev97,erdos1972linear,lev2024large,olson1984sum}. Unfortunately, as we shall discuss later in Subsection~\ref{subsec:tec-review}, the contradiction-based arguments in the proofs of these two theorems make it impossible to directly modify S{\'a}rk{\"o}zy or Lev's proof of Theorem~\ref{thm:ka-old}. We establish new techniques that can bypass these two theorems and yield a similar result. Our new techniques are fully constructive, and can return an arithmetic progression claimed in Theorem~\ref{thm:ka-old} together with a witness in near-linear time. Our technique does not lead to a fully constructive proof for Mann's theorem and Kneser's theorem, but it implies that when the two theorems are applied in a "dense" case, e.g., in $k$-fold sumset where $k$ is large, or subset sums where the cardinality of a subset is large, then there is an alternative way to obtain efficient constructiveness.


Below we present our constructive version of Theorem~\ref{thm:ka-old}. Recall that an arithmetic progression $\{s\} + \{0, d, \ldots, \ell d\}$ can be compactly represented by the tuple $(s, \ell, d)$. 

\begin{restatable}{theorem}{thmka}
\label{thm:ka}
    Let $A \subseteq \mathbb{Z}[0,m]$ be a set of $n$ integers. Assume that $0 \in A$ and that $\gcd(A) = 1$.  Let $k$ be a positive integer. Assume that
    \[
            n \geq \frac{m + 1}{k}.
    \]  
    In $O(n\log m)$ time, we can compute an arithmetic progression 
    \[
        \{s\} + \{0, 1, 2, \ldots, m\} \subseteq 332kA,
    \]
    and obtain a randomized witness with $O(\min\{\frac{m}{n},n\} \log n + \log m)$ expected query time.\footnote{The witness is a Las Vegas algorithm. It always answers the query correctly, while the query time is randomized. }
\end{restatable}

\noindent\textbf{Remark.} Here the data structure of the witness can be computed deterministically in $O(n\log m)$. Randomization is only necessary for returning the witness per query. Note that while a witness for some $t\in 332kA$ can consist of $\Theta(k)=\Theta(\frac{m}{n})$ elements, in the case of $n\le \frac{m}{n}$ the witness can be compactly encoded (i.e., by specifying the number of each distinct element used in the witness), so $O(\min\{\frac{m}{n},n\} \log n + \log m)$ time is sufficient.

 The query time of the witness is the best possible up to a logarithmic factor since a solution for $z \in kA$ may consist of $\min\{\frac{m}{n}, n\}$ distinct integers. Such a query time is highly non-trivial: in general, it remains open whether a solution for $z \in A + A + A$ (that is, find a solution for the 3SUM problem) can be found in truly sub-quadratic time. The main difference between our problem and 3SUM is that we are considering a "dense" case where $k = \Omega(\frac{m}{n})$. In a dense case, one can expect that there are many solutions for $a \in kA$ and that some of them can be found efficiently.

\subsubsection{Finite Addition Theorem II -- Long Arithmetic Progression in Subset Sums}
Building upon Theorem~\ref{thm:ka-old}, S{\'a}rk{\"o}zy~\cite{Sar94} further extended the finite addition theorem for Subset Sum. This extension has a similar flavor as Theorem~\ref{thm:ka-old}, which states that when $|A|\geq c\sqrt{m\log m}$, for some positive constant $c$, then the set of its subset sums, $\mathcal{S}(A)$, contains an arithmetic progression of length $m$. Almost 10 years later, Lev~\cite{Lev03} improved the constant. Below we present Lev's version. 

\begin{theorem}[Finite Addition Theorem II, {\cite[Theorem 3]{Lev03}}]\label{thm:ss-old}
    Let $A \subseteq \mathbb{Z}[1, m]$ be a set of $n$ integers. For any integer $\ell$ with  
    \[
       4m \leq  \ell \leq \frac{n^2}{12 \log (4m/n)},
    \]
    there exists an arithmetic progression
    \[
        \{s\} + \{0, d, 2d, \ldots, \ell d\} \subseteq \mathcal{S}(A)
    \]
    with $d \leq \frac{4m}{n}$. Moreover, each term of the arithmetic progression is a sum of at most $\frac{6\ell}{n}$ distinct elements of $A$.
\end{theorem}

It is worth mentioning that Szemerédi and Vu~\cite{SV06a} later obtained a sharper version of Theorem~\ref{thm:ss-old} that only requires $n=\Omega(\sqrt{m})$ instead of $n=\Omega(\sqrt{m\log m})$. Their proof is completely different and heavily relies on an "inversion argument" which is non-constructive. It is far from clear how to make Szemerédi and Vu's proof constructive. Nevertheless, we can obtain a constructive version of Theorem~\ref{thm:ss-old}, as presented below. 

\begin{restatable}{theorem}{thmss}
\label{thm:ss}
    Let $A \subseteq \mathbb{Z}[1,m]$ be a set of $n$ integers. For any integer $\ell$ with 
    \[
        m \leq \ell \leq \frac{n^2}{5 \times 10^8\log (2n)},
    \]
    in $O(n \log n)$ time, we can compute a subset $A' \subseteq A$ with $|A'| \leq \frac{3 \times 10^5 \ell\log n}{n}$, an arithmetic progression
    \[
        \{s\} + \{0, d, 2d, \ldots, \ell d\} \subseteq \mathcal{S}(A')
    \]
    with $d \leq \frac{7m}{n}$, and obtain a randomized witness with $O(\frac{\ell \log n}{n})$ expected query time.
\end{restatable}
We remark that once we have Theorem~\ref{thm:ka}, S{\'a}rk{\"o}zy or Lev's proof for Theorem~\ref{thm:ss-old} can be made constructive directly, except that the time for returning an arithmetic progression can be as large as $\Omega(n^2)$ (see Subsection~\ref{subsec:tech-ss} for a detailed discussion). The merit of Theorem~\ref{thm:ss} is that it runs in near-linear time. 

We compare Theorem~\ref{thm:ss} with Theorem~\ref{thm:ss-old}. In Theorem~\ref{thm:ss-old}, the long arithmetic progression in $\mathcal{S}(A)$ may involve every integer in $A$. In contrast, our theorem shows that we can extract a small subset $A'$ of cardinality $O(\frac{\ell \log n}{n})$, whose subset sums already guarantee a long arithmetic progression.  For example, if $n = cm$ for some constant $c$, then a subset $A'$ of cardinality $O(\log n)$ is sufficient to guarantee that $\mathcal{S}(A')$ contains an arithmetic progression of length $m$. In some sense, $A'$ is a "coreset" of $A$ which is sparse but preserves the nice additive structure that yields a long arithmetic progression. The "denser" $A$ is, the "sparser" its coreset can be. To our best knowledge, such a coreset property has not been known before, nor can it be derived from the previous proofs. It may be of independent interest. We also remark that the cardinality of the coreset is the best possible (up to a constant factor) since in the case of $n = cm$, for $\mathcal{S}(A')$ to have $m$ distinct elements, $A'$ must have at least $\log m$ integers. 

It is worth mentioning that Galil and Margalit~\cite{GM91} also present a constructive version of Theorem~\ref{thm:ss-old}, and obtained a witness in $\widetilde{O}(m)$ time. The major advantage of our Theorem~\ref{thmss} is that it only requires $\widetilde{O}(n)$-time, while $m$ can be as large as $\widetilde{O}(n^2)$. In terms of techniques, the main advantage of our method is that it avoids the use of FFT (Fast Fourier Transform), which will be discussed later in Subsection~\ref{subsec:tech}.

\subsubsection{Applications in Unbounded Subset Sum and Dense Subset Sum}
With our constructive versions of finite addition theorems, we are able to close the gap between the time complexity of the decision version and the search version of several fundamental problems.

\paragraph{Unbounded Subset Sum.} Theorem~\ref{thm:ka} implies the following result for Unbounded Subset Sum.

\begin{restatable}{theorem}{thmunbounded}\label{thm:unbounded-subset-sum}
    Given $n$ positive integers $a_1 < a_2 < \cdots < a_n$ with $\gcd(a_1,a_2,\cdots,a_n)=1$ and an integer $t$, if 
    \[
        t \geq 333\lceil\frac{a_n}{n-1}\rceil a_{n-1},
    \]
    then in $O(n \log a_n)$ expected time, we can obtain $n$ non-negative integers $x_1, \ldots, x_n$ such that 
    \[
        t = a_1x_1 + \cdots + a_n x_n.
    \]
\end{restatable}
Note that the threshold of $333\lceil\frac{a_n}{n-1}\rceil a_{n-1}$ in our theorem is within an $O(1)$ factor of the Erd\H{o}s-Graham bound (See~\eqref{eq:Frobenius}). Moreover, the running time is almost linear. 

\paragraph{Dense Subset Sum} Using Theorem~\ref{thm:ss}, we can adapt Bringmann and Wellnitz's $\tilde{O}(n)$-time algorithm~\cite{BW21} for Dense Subset Sum so that it not only solves the decision version, but also returns a subset that sums to $t$ when the answer to the decision problem is "yes". Since the theorem is quite involved, we only present an informal version here. One may refer to Theorem~\ref{thm:dense} for a formal version of the theorem.

\begin{theorem}\label{thm:dense-ss-set-informal}
    Given a set of $A$ of $n$ positive integers and a target $t \leq \frac{1}{2} \sum_{a \in A}a$, if \footnote{Here $f \gg g$ means $f$ is at least a polylogarithmic factor larger than $g$. }
    \[
        t \gg \max(A)\cdot \sum_{a \in A}a \cdot \frac{1}{n^2}, 
    \]
    then in $\tilde{O}(n)$ expected time, we can decide whether there is some subset $A' \subseteq A$ that sums to $t$, and returns such a set when it exists. 
\end{theorem}
We remark that the above result can be extended to multi-sets as in~\cite{BW21}, but we leave it to Appendix~\ref{apsec:dense}.

\subsection{Technique Overview}\label{subsec:tech}
\subsubsection{Overview of S{\'a}rk{\"o}zy's Proof for Finite Addition Theorem I (Theorem~\ref{thm:ka-old})}\label{subsec:tec-review}

S{\'a}rk{\"o}zy's proof can be decomposed into two parts. The first part guarantees the existence of an arithmetic progression whose length is moderately large (but can be potentially much smaller than $m$). This part builds upon the famous Mann's theorem. The second part prolongs the arithmetic progression found in the first step. 
Let $g$ be the common difference of the arithmetic progression found in the first step. The second part targets showing the existence of some $g'$ which divides $g$, such that the "gaps" of distance $g$ in the arithmetic progression can now be subdivided into "gaps" of distance $g'$, thus blowing up the length of the arithmetic progression by $g/g'$ times. The existence of such a $g'$ is guaranteed by Kneser's theorem. We remark that Lev's proof~\cite{Lev97} follows a similar approach, despite that it provides a much fine-grained analysis in several places. 

Towards a constructive theorem, we thus need to revisit Mann's theorem and Kneser's theorem. 

\paragraph{Mann's theorem} In his seminal paper \cite{schnirelmann1933additive}, Schnirelmann initiated the study of {\it Schnirelmann density}, which builds the foundation of many subsequent studies in additive combinatorics\footnote{The definition of Schnirelmann density is on infinite sets, but relevant results can be adapted to finite sets. As this paper focuses on finite sets, we shall present relevant concepts and equations in a form that only involves finite sets. }. 

Let $A \subseteq \mathbb{Z}$ be a set of integers and $0\in A$.  The density of $A$ is defined as follows. 
\[
    \rho_m(A) = \min_{z' \in \mathbb{Z}[0,m]}\frac{A \cap [1, z']}{z'}.
\]

Schnirelmann showed that
\begin{eqnarray}\label{eq:schnirelmann}
    \rho_m(A + B) \ge  \rho_m(A) + \rho_m(B) - \rho_m(A)\rho_m(B).
\end{eqnarray}

Schnirelmann and Landau \cite{khinchin1952three} conjectured that Eq~\eqref{eq:schnirelmann} can be sharpened to the following.
\begin{eqnarray}\label{eq:mann}
 \rho_z(A+B)\ge \min\{\rho_z(A)+\rho_z(B),1\}.
 \end{eqnarray}

Eq~\eqref{eq:mann} was proved later by Mann~\cite{man42} and is now known as Mann's theorem.
 
 Let us compare Eq~\eqref{eq:schnirelmann} and Eq~\eqref{eq:mann}. Eq~\eqref{eq:mann} is obviously sharper. However, since its proof is based on a minimal counterexample, it is non-constructive. In particular, suppose we know that $t\in A+B$ and want to search for $(a,b)\in A\times B$ with $a+b=t$, then Mann's proof does not help. 

Interestingly, Schnirelmann's proof is constructive, albeit that it yields a weaker bound. In particular, Schnirelmann proved Eq~\eqref{eq:mann} by considering only the sums that can be obtained in a specific greedy way from $A$ and $B$. In other words, knowing that $t\in A+B$ from Schnirelmann's proof allows us to efficiently search for $(a,b)\in A\times B$ with $a+b=t$. It points out a way to replace the usage of Mann's theorem with Schnirelmann's theorem. Details will be elaborated in Subsection~\ref{subsec:techka}.

\paragraph{Kneser's theorem} In the 1950s, Kneser~\cite{kneser1953abschatzung} established a famous theorem in additive combinatorics. It is commonly stated in terms of abelian groups, but as our work focuses on integer sets, we shall follow its statement on integer sequences (\cite{HR83}, page 52, Theorem 16), and the following is an essentially equivalent statement (see also Lemma 4 of~\cite{Sar89}). 

Let $\mathbb{N}=\{0,1,2,\cdots\}$ denote the set of nonnegative integers. 

\begin{theorem}[Kneser's theorem]~\label{thm:kneser}
	Assume that $h\in\mathbb{N}$, $g\in\mathbb{N}$, $A\subset \mathbb{N}$, and $A$ is the union of $\gamma$ distinct residue class modulo $g$: $A=\bigcup_{i=1}^{\gamma}\{a_i,a_i+g,a_i+2g,\cdots\}$, where $a_i\not\equiv a_j (\mod g)$ for $i\neq j$. Then there is a divisor $g'$ of $g$ such that $hA$ is the union of $\gamma'$ distinct residue classes modulo $g'$: $hA=\bigcup_{i=1}^{\gamma'}\{e_i,e_i+g,e_i+2g,\cdots\}$ where
\begin{eqnarray}\label{eq:kneser}
 \frac{\gamma'}{g'}\geq h\frac{\gamma}{g}-\frac{h-1}{g'}.
\end{eqnarray}	

\end{theorem} 
 Note that since $a_i$'s represent different residues modulo $g$, we may assume without loss of generality that $a_i\in\mathbb{Z}[0,g-1]$.

 
 The above theorem provides a structural characterization of modular $h$-fold sumsets for arbitrary $h$. 
 In particular, S{\'a}rk{\"o}zy~\cite{Sar89} utilized Theorem~\ref{thm:kneser} to show the following Lemma.
 
 \begin{lemma} [Kneser's theorem in the dense case]\label{lemma:sar-lemma} 
 	Assume that $g\in\mathbb{N}$, $A\subset \mathbb{N}$, $d\neq d' (\mod g)$ for $d\in A$, $d'\in A$, and let $|A|=\gamma$. Then there is a divisor $g'$ of $g$ and a number $h\in\mathbb{N}$ such that $g'\le \lfloor 2g/\gamma\rfloor$, $h\le 2\lfloor2g/\gamma\rfloor$, and for each of $i=1,2,\cdots,g/g'$, there is a $z_i\in hA$ with $z_i\equiv ig' (\mod g)$.
 \end{lemma}

 Lemma~\ref{lemma:sar-lemma} can be viewed as a consequence of Theorem~\ref{thm:kneser} in a "dense" case where $h$ is moderately large (and $g'$ can thus be moderately small via Eq~\eqref{eq:kneser}). We remark that Erd\H{o}s and Graham also obtained a similar result through Theorem~\ref{thm:kneser} for the Frobenius number~\cite{erdos1972linear}. 
 

 As we mentioned before, the second part of S{\'a}rk{\"o}zy's proof uses some $g'$ to blow up the arithmetic progression obtained in the first part, and it is exactly the $g'$ stated in Lemma~\ref{lemma:sar-lemma}. Unfortunately, Kneser's theorem as well as Lemma~\ref{lemma:sar-lemma} does not give the value of $g'$. Given that no known algorithms can solve integer factorization in polynomial time, even trying all divisors of $g$ by bruteforce requires a running time super-polynomial in $\log g$.

  It is far from clear whether there exists a polynomial time algorithm for computing the $g'$ stated in Theorem~\ref{thm:kneser}. Nevertheless, we shall present a method that yields a near-linear time for computing the $g'$ in the dense case, 
  and returning a witness.


\subsubsection{Overview of Our Technical Contribution in Theorem~\ref{thm:ka}}\label{subsec:techka}
\paragraph{Bypassing Mann's theorem - Schnirelmann's proof strikes back.} 

As we mentioned,  Eq~\eqref{eq:schnirelmann} was showed by Schnirelmann by considering a specific subset $C\subseteq A+B$. In particular, Schnirelmann's proof implies that every $c\in C$ can be represented as $a+b$ where $a$ is the largest integer in $A$ that is smaller than or equal to $c$. Inspired by Schnirelmann's proof, we introduce a new operation $\oplus$ between two sets $A$ and $B$, called greedy sumset. The greedy sumset $A\oplus B$ returns $C$, which corresponds exactly to the proof of Schnirelmann. Consequently, knowing that $c\in A\oplus B$, searching for $a\in A$ and $ b\in B$ with $a+b=c$ can be easily done in logarithmic time provided that $A$ is presorted.


Our plan is to replace the sumset computation in S{\'a}rk{\"o}zy's proof with greedy sumset, with the hope that Finite Addition Theorems I and II still hold (maybe at the cost of blowing up certain parameters by $O(1)$ times). Towards this goal, let us compare $A+B$ with $A\oplus B$. It is worth noticing that the size of  $A\oplus B$ can be significantly smaller than $A+B$. A crucial observation is that if one only uses the lower bound of Eq~\eqref{eq:mann} to estimate the density of $A+B$ (as is the case in many prior papers in additive combinatorics), then replacing $A+B$ with $A\oplus B$ does not incur a significant loss. The loss becomes particularly marginal in a dense scenario, namely the $k$-fold sum with large $k$. Indeed, when $k=\Omega(m/|A|)$, simple calculations via Eq~\eqref{eq:mann} and Eq~\eqref{eq:schnirelmann} reveal that the density of $kA$ and $k\otimes A:=A\oplus A\oplus\cdots\oplus A$ only differ by an arbitrary small factor.


The notion of greedy sumset, together with our analysis for replacing sumset with greedy sumset, may be of separate interest for obtaining a constructive version of other additive combinatorics results.


\paragraph{Bypassing Kneser's theorem - an iterative AP-augmentation framework.} Recall that S{\'a}rk{\"o}zy utilized Mann's theorem to obtain a moderately long arithmetic progression of some common difference $g$, and then used Kneser's theorem (more precisely, Lemma~\ref{lemma:sar-lemma}) to find a sufficiently small divisor $g'$ of $g$ to blow up the arithmetic progression by $g/g'$ times. Roughly speaking, we may interpret the blowing-up procedure as follows: if we find a long arithmetic progression of common difference $g$ in $k_1A$, and a short arithmetic progression of common difference $g'$ in $k_2A \,\, (\mod g)$, then we can find a long arithmetic progression of common difference $g'$ in $(k_1+k_2)A$. 


Unfortunately, the $g'$ stated in Lemma~\ref{lemma:sar-lemma} cannot be computed directly. Instead, we establish an augmentation framework that iteratively finds some divisor $g''$ of $g$ (which can be large) and blows up the current arithmetic progression by $g/g''$ times. After $O(\log g)$ iterations, we obtain the desired long arithmetic progression. This idea can be easily implemented to obtain a weaker result, namely $(k\log g) A$ contains an arithmetic progression of length $m$. To make sure that $k A$ is sufficient, the cost incurred by each iteration has to be carefully balanced. In particular, each iteration needs to be "patched" with another augmentation procedure so that we can reduce the average "cost" of increasing the length of the arithmetic progression.

\subsubsection{Overview of Our Technical Contribution in Theorem~\ref{thm:ss}}\label{subsec:tech-ss}

Moving from $k$-fold sumset to Subset Sum, the major challenge is that in Subset Sum each integer of $A$ can only be used once. To overcome this difficulty,  S{\'a}rk{\"o}zy~\cite{Sar94}   considered $A+A$ and showed that there exists some $B\subseteq A+A$ such that every $b \in  B$ can be represented as a sum of two integers of $A$ in many disjoint ways. In other words, each integer $b \in B$ can be used multiple times. Consequently, finding arithmetic progressions in subset sums of $A$ can be transformed to finding arithmetic progressions in $k$-fold sumset in $B$. It is worth mentioning that this general idea of creating multiplicity via $A+A$ has been adopted in many prior works in additive combinatorics, e.g.,~\cite{Lev03,SV06a,SV05}.

We follow the general idea of transforming Subset Sum to $k$-fold sumset, which requires creating some suitable set from $A$ that is sufficiently dense. However, we cannot follow the particular approach above because we are targeting a linear-time algorithm, but the above-mentioned $B\subseteq A+A$ may have a cardinality of $\Omega(n^2)$.

We consider $ A - A$ instead. The advantage of working on $A-A$ is that we can find a suitable subset $G\subseteq A-A$ that contains many integers upper bounded by  $O(\frac{m}{n})$, in other words, $G$ is sufficiently dense. The disadvantage is that because its integers are small, using $G$ we can only obtain a short arithmetic progression $P$ of length roughly $\frac{m}{n}$. The challenging part is to augment this arithmetic progression. To this end, we use our iterative augmentation framework. The augmentation in Subset Sum is more complicated. Roughly speaking, we iteratively make the current arithmetic progression longer. In each iteration, we blow up the current arithmetic progression by: (i) either finding a divisor of the current common difference (and argue in a similar way as we do for Theorem~\ref{thm:ka}); (ii) or finding sufficiently many integers in $A-A$ that are multiples of the current common difference, and are not too small nor large. We can show that these integers can be used to significantly increase the length of the current arithmetic progression (without changing the common difference). The existence of the integers in $A-A$ satisfying the above-claimed property is non-trivial, and leads to the existence of a coreset $A'$ claimed in Theorem~\ref{thm:ss}.   

\smallskip
\noindent\textbf{Comparison with the work of Galil and Margalit~\cite{GM91}.} We remark that Galil and Margalit~\cite{GM91} also gave a constructive proof that leads to a $\widetilde{O}(m)$-time algorithm for constructing the witness. Their overall framework resembles us in the sense that it also constructs a short arithmetic progression and then iteratively augment it. However, their augmentation crucially relies on computing $A-A$ first via FFT. FFT not only returns $A-A$, but also counts the number of pairs $(a,b)\in A\times A$ such that $a-b=c$ for every $c\in A-A$. The method of Galil and Margalit relies on the count numbers to strategically pick integers in $A$ to augment an arithmetic progression. Unfortunately, performing FFT requires $\widetilde{O}(m)$-time. To obtain a near-linear time algorithm, the augmentation procedure cannot be guided by count numbers, which is the major challenge.




\subsection{Other Related Work}

Both Subset Sum and Unbounded Subset Sum are fundamental NP-hard problems in theoretical computer science~\cite{Kar72}. There is a long history of study on pseudopolynomial time algorithms for Subset Sum~\cite{Bel57,Pis99,pisinger2003dynamic,KX17,KX18,Bri17,JW19,PRW21,chaimovich1999new,chaimovich1989solving,GM91,BW21,chen2024improved}, and on pseudopolynomial time algorithms for Unbounded Subset Sum~\cite{Bri17,jansen2023integer,hansen1996testing,klein2022fine}. A phenomenon that is unique to Unbounded Subset Sum is the Frobenius number.  
Computing the value of the Frobenius number is NP-hard~\cite{klein2022fine}. 
For $n=2$ the Frobenius number is given by $F(a_1,a_2)=a_1a_2-a_1-a_2$~\cite{sylvester1882subvariants}; for $n=3$ relatively sharp estimates are given in~\cite{beck2004refined,ustinov2009solution}. For general $n$ the best-known is given by Erd\H{o}s and Graham~\cite{erdos1972linear}. 



It is remarkable that the successful application of many additive combinatorial results in algorithmic design in recent years has also motivated the effort in the search for a constructive version of these results, including, e.g., the algorithm for Balog-Szemeredi-Gowers (BSG)
theorem by  Chan and Lewenstein~\cite{chan2015clustered}, a constructive version of Ruzsa’s covering
lemma by Abboud, Bringmann, and Fischer~\cite{abboud2023stronger}, a constructive version of Freimann's theorem by Randolph and W{\k{e}}grzycki~\cite{randolph2024parameterized}.

\subsection{Paper Outline}
In Section~\ref{sec:pre}, we introduce necessary terminologies and preliminaries. Section~\ref{sec:density} defines the density of a set $A$ and shows that if the density is $\Omega(1/k)$, then $kA$ contains a long arithmetic progression. Based on this result, Section~\ref{sec:ka} proves that if $n \geq \Omega(m/k)$, then $kA$ contains a long arithmetic progression. Finally, Section~\ref{sec:ss} proves if $n \geq \Omega(\sqrt{m\log m})$, then $\mathcal{S}(A)$ contains a long arithmetic progression.  We give two applications of our results in Section~\ref{sec:application} and conclude the paper in Section \ref{sec:conclude}. 

\section{Notations}\label{sec:pre}
Throughout the paper, all logarithms are based 2. Given two integers $a, b$, the gap between them is defined to be $|a - b|$.

Let $A$ be a set of integers.  Let $x$ and $y$ be two real numbers. We use $A[x,y]$ to denote the set of integers in $A$ that are between $x$ and $y$. That is,
\[
    A[x, y] = \{a \in A : x \leq a \leq y\}.
\]

Given two integers $a$ and $b$, we use $\gcd(a, b)$ to denote their greatest common divisor. $\gcd(a , b)$ can be computed in $O(\log \frac{\min\{a,b\}}{\gcd(a,b)})$ time by the Euclidean algorithm~\cite{shallit1994origins}.
Given a set $A$ of integers, we use $\gcd(A)$ to denote the greatest common divisor of all the integers in $A$.

Let $d$ be a positive integer.  We say $d'$ is a proper divisor of $d'$ if $1\leq d' < d$. Note that the integer $1$ has no proper divisor.

\section{Long Arithmetic Progressions by Density}\label{sec:density}
We shall show how to obtain an arithmetic progression by considering the density of a set of integers. 

Let $A$ be a set of $n$ integers. Let $z$ be a positive integer. The density of $A$ over the interval $[1, z]$ is defined to be
\[
    \rho_z(A) = \min_{1 \leq z' \leq z} \frac{|A[1, z']|}{z'}.
\]

\noindent\textbf{Remark.} The definition of density is very sensitive to small numbers within $A$. Indeed, if $1\not\in A$, then $\rho_z(A)=0$. Therefore, we will typically assume that $1\in A$ to guarantee a nonzero density.

Assume that $0\in A$. Dyson's theorem~\cite{Dys45} (and also Mann's theorem~\cite{man42}) implies that
\[
    \rho_z(kA) \ge \min\{1, k\rho_z(A)\}.
\]
Therefore, when $k$ is large enough, the density of $kA$ over $[1, z]$ will be 1, and as a result, $kA$ will contain all the integers in $\mathbb{Z}[1, z]$. 
\begin{lemma} [\cite{man42,Dys45}] \label{lem:ap-by-density}
    Let $A \subseteq \mathbb{Z}[0,m]$ be a set of integers. Assume that $\{0,1\} \in A$.  For any integer $k$ with $k \geq 1/\rho_m(A)$, we have that 
    \[
        \{0,1, \ldots, m\} \subseteq kA.
    \]
\end{lemma}
Both Dyson's theorem and Mann's theorem are, however, not constructive. As a consequence, to find a solution for even a single term, one has to compute $kA$ explicitly, which can take $O(m\log m\log k)$ time in the worst case. 

This section presents a constructive proof for Lemma~\ref{lem:ap-by-density}.  Following the construction in our proof, we can also obtain an efficient witness. We summarize the main result of this section by the following lemma. 

\begin{restatable}{lemma}{lemconsapbydensity}
\label{lem:ap-by-density-cons}
    Let $A \subseteq \mathbb{Z}[0,m]$ be a set of $n$ integers. Assume that $\{0,1\} \subseteq A$.  Let $k$ be a positive integer with
    \(
            k \geq 2/\rho_m(A).
    \)  
    In $O(n \log n)$ time, we can construct a witness with $O(k \log n)$ expected query time for 
    \[
        \{0, 1, \ldots, m\} \subseteq 2kA.
    \]
\end{restatable}

\noindent\textbf{Remark.} Here the running time of $O(n \log n)$ is mainly due to the fact that we need to sort the integers within $A$. If $A$ is presorted, then the witness can be constructed in $O(n)$ time.

\subsection{Greedy Sumsets}
In general, finding a solution for $z \in kA$ can be expensive. For example, it remains open whether a solution for $z \in A + A + A$ can be founded in truly subquadratic time. Fortunately, now we are dealing only with the cases where $k$ is large, and when $k$ is large, the number of solutions for $z \in kA$ is also large.  We observe from Schnirelmann's seminal work~\cite{schnirelmann1933additive} that when $k$ is large enough, there must be a solution that can be obtained greedily.  To formalize this observation, we define the notion of greedy sumset.

\begin{definition}\label{def:greedy-sumset}
    Let $A$ and $B$ be two sets of integers.  Assume that $a_1 < a_2 \ldots < a_n$ are the elements of $A$. Let $a_{n+1} = \infty$. We define the greedy sumset of $A$ and $B$ as follows.
    \[
        A \oplus B = \{\text{$a_i + b$ : $a_i \in A$ and $b \in B$ and $0 \leq b < a_{i+1} - a_i$}\}
    \]
\end{definition}

It is easy to see that $A \oplus B$ is a subset of $A + B$.  Basically, $A \oplus B$ is the set of sums in $A + B$ that can be represented greedily.  In particular, for each $z \in A \oplus B$, a solution $(a,b)$ for $z\in A + B$ can be founded by greedily choosing $a$ to be the largest integer in $A$ not exceeding $z$.  
\begin{lemma}\label{lem:greedy-sumset-property}
    Let $A$ and $B$ be two sets of integers. Let $z$ be an integer. Let $a$ be the largest integer in $A$ not exceeding $z$. Then $z \in A \oplus B$ if and only if $z - a \in B$.
\end{lemma}
\begin{proof}
    We first prove the if part. Suppose that $z - a \in B$. Let $a'$ be the successor of $a$ in $A$. If the successor does not exist, we define $a' = \infty$. To prove that $z \in A \oplus B$, it suffices to show that $z - a < a' - a$. This is straightforward since $a \leq z <  a'$.

    Next we prove the only if part. Suppose that $z \in A\oplus B$. By definition of greedy sumset, there must exist $a' \in A$ such that $z - a' \in B$ and that $0\leq z - a' < a'' - a'$ where $a''$ is the successor of $a'$ in $A$ (In case that the successor does not exist, $a'' = \infty$).  This implies that 
    \[
        a' \leq z < a''.
    \]
    Therefore, $a'$ must be the largest integer in $A$ not exceeding $z$.
\end{proof}

Using the above property, given any integer $z$, we can determine whether $z \in A \oplus B$ in $O(\log n)$ time. Moreover, if $z \in A \oplus B$, we can construct a solution for $z \in A + B$ in $O(\log n)$ time. 
\begin{lemma}\label{lem:greey-sumset-solution}
    Let $A$ and $B$ be two sorted sets of $n$ non-negative integers. In $O(n)$ time, we can construct a data structure that is able to answer the following query in $O(\log n)$ time: given any integer $z$, is $z \in A \oplus B$? If yes, return a solution for $z \in A + B$.
\end{lemma}
\begin{proof}
    The data structure simply stores $A$ and $B$. To answer the query, we first find the largest $a\in A$ not exceeding $z$.  Then we check whether $z-a\in B$ or not. If yes, we return $(a, z-a)$, which is a solution for $z \in A + B$. If no, then $z \notin A\oplus B$.  Both operations can be done in $O(\log n)$ time using binary search.  The correctness follows by Lemma~\ref{lem:greedy-sumset-property}.
\end{proof}

Compared with $A + B$, although $A \oplus B $ contains only the sums that are obtained greedily, its density is still significant. We remark that the proof of the following lemma is essentially the same as that of Schnirelmann's theorem (see Chapter I of \cite{HR83} for reference).

\begin{lemma}\label{lem:two-greedy-density}
    Let $A$ and $B$ be two sets of integers. Assume that $1 \in A$ and $0 \in B$. Let $z$ be a positive integer. Then
    \[
        \rho_z(A \oplus B) \geq \rho_z(A) + \rho_z(B) - \rho_z(A)\rho_z(B).
    \]
\end{lemma}
\begin{proof}
    To prove the lemma, it suffices to show that for any $z' \in [1, z]$,
    \[
        \frac{(A \oplus B)[1, z']}{z'} \geq  \rho_z(A) + \rho_z(B) - \rho_z(A)\rho_z(B).
    \]  

    Take an arbitrary $z' \in \mathbb{Z}[1,z]$. Label the elements of $A[1, z']$ as $\{a_1, \ldots, a_n\}$ in increasing order.  Note that $a_1 = 1$ as $1 \in A$. For simplicity, we define $a_{n+1} = z' + 1$.  The set $\mathbb{Z}[1, z']$ can be partitioned into $\bigcup_{i=1}^n \mathbb{Z}[a_i, a_{i+1}-1]$. Consider an integer $y \in \mathbb{Z}[a_i, a_{i+1}-1]$. If $y = a_i$, then $y \in A\oplus B$ as $0 \in B$. If $y = a_i + s$ for some $s \in \mathbb{Z}[1, a_{i+1} - a_i - 1]$, then $y \in A\oplus B$ if and only if $s \in B$. Therefore,
    \[
        |(A\oplus B)[a_i, a_{i+1} + 1]| = 1 + |B[1, a_{i+1} - a_i - 1]| \geq 1 + \rho_z(B) \cdot (a_{i+1} - a_i - 1).
    \]
    The last inequality is due to the definition of $\rho_z(B)$. Taking sum over $i$, we have that
    \begin{align*}
        |(A \oplus B)[1, z']| &\geq n + \rho_z(B) \sum_{i=1}^n (a_{i+1} - a_i - 1)\\ 
                            & = n + \rho_z(B)\cdot (z' - n)\\
                            & = \rho_z(B) \cdot z' + (1 - \rho_z(B))n\\
                            & \geq \rho_z(B) \cdot z' + (1 - \rho_z(B))\rho_z(A) \cdot z'\\
                            & = z' \cdot (\rho_z(A) + \rho_z(B) - \rho_z(A)\rho_z(B)).      \qedhere
    \end{align*}
\end{proof}

Next we extend Lemma~\ref{lem:greey-sumset-solution} and Lemma~\ref{lem:two-greedy-density} to $k$-fold greedy sumsets. Given a set $A$ of integers, its $k$-fold greedy sumset is defined as follows: $k\otimes  A = \{0\}$ when $k=0$, and $(k+1)\otimes  A = A\oplus (k\otimes  A)$ for all $k \geq 1$.  (Note that the $\oplus$ operation is not commutative nor associative. Therefore, $A\oplus (k\otimes  A) \neq (k\otimes  A) \oplus A$.)

\begin{lemma}\label{lem:k-greey-sumset-solution}
    Let $A$ be a sorted set of $n$ non-negative integers. In $O(n)$ time, we can obtain a data structure that is able to answer the following query in $O(k \log n)$ time:  given two non-negative integers $z$ and $k$, is $z \in k\otimes A$? If yes, also return a solution for $z \in kA$.
\end{lemma}
\begin{proof}
    Our data structure simply stores $A$. To answer a query for $z$ and $k$, we iterative compute $a_i$ for $i \in \{1, \ldots, k\}$ as follows: let $a_i$ be the largest integer in $A$ not exceeding $z$, and subtract $a_i$ from $z$.  If we successfully find all $a_1, \ldots, a_k$, and they sum to $z$, then $z \in k\otimes A$ and we have a solution for $z \in kA$. Otherwise, we conclude that $z \notin k\otimes A$.

    The correctness easily follows from induction on $k$, where the base case $k=2$ is due to Lemma~\ref{lem:greedy-sumset-property}. We omit the details of the proof. We have at most $k$ iterations, and each takes $O(\log n)$ time, so the query time is $O(k \log n)$.
\end{proof}

\begin{lemma}\label{lem:k-greedy-density}
    Let $A$ be a set of non-negative integers. Assume that $\{0,1\} \subseteq A$. Let $z$ be a positive integer.  Then
    \[
        \rho_z(k\otimes A) \geq 1 - (1 - \rho_z(A))^k.
    \]
\end{lemma}
\begin{proof}
    When $k = 1$, the lemma holds straightforwardly. Suppose that the lemma holds for $k$. We show that it holds for $k+1$. Since $(k+1)\otimes  A = A\oplus (k \otimes  A)$, by Lemma~\ref{lem:two-greedy-density},
    \begin{align*}
        \rho_z((k+1)\otimes  A) &\geq \rho_z(A) + \rho_z(k\otimes A) - \rho_z(A)\rho_z(k\otimes  A)\\
            &= 1 - (1 - \rho_z(A))(1 - \rho_z(k\otimes A))\\
            & \geq 1 - (1 - \rho_z(A))\cdot (1 - \rho_z(A))^k\\
            &= 1 - (1 - \rho_z(A))^{k+1}.
    \end{align*}
    The second inequality is due to the inductive hypothesis.
\end{proof}

\subsection{Long Arithmetic Progressions}
By Lemma~\ref{lem:k-greedy-density}, we have that $\rho_m(k\otimes A) = 1$ for $k \geq \frac{\ln m}{\rho_m(A)}$. Therefore,
\[
    \{0, 1, \ldots, m\} \subseteq k\otimes  A.
\]
Then, by Lemma~\ref{lem:k-greey-sumset-solution}, in $O(n\log n)$ time, we can obtain an efficient witness with $O(k \log n)$ time.  

Below we shall shave the logarithmic factor in the threshold for $k$. We first add up $A$ (greedily) for $O(\frac{1}{\rho_m(A)})$ times so that the density becomes a constant. 

\begin{lemma}\label{lem:constant-density}
    Let $A \subseteq \mathbb{Z}[0,m]$ be a set of integers. Assume that $\{0,1\}\subseteq A$. Let $k = \lceil \frac{2}{\rho_{m}(A)} \rceil$. Then 
    \[
        \rho_m(k\otimes  A) \geq \frac{3}{4}.
    \]
\end{lemma}
\begin{proof}
    Let $\rho = \rho_{m}(A)$. Note that 
    \[
        k \geq \frac{2}{\rho} \geq \frac{2}{ - \ln (1 - \rho)}.
    \]
    By Lemma~\ref{lem:k-greedy-density}, 
    \[
        \rho_m(k\otimes  A) \geq 1 - (1 - \rho)^k \geq 1 - \frac{1}{e^2} > \frac{3}{4}. \qedhere
    \]
\end{proof}

Next we show that when a set $A$ has large density over an interval $[1, m]$, then $\{1, \ldots, m\} \in A + A$. Moreover, for any $z\in \{1, \ldots, m\}$, we can guess a representation $z=a+(z-a)$ by randomly picking $a\in \mathbb{Z}[1,z]$, and the guess is correct with a constant probability.
\begin{lemma}\label{lem:sample-solution}
    Let $A$ be a set of integers. Assume that $0 \in A$ and that $\rho_m(A) \geq \frac{3}{4}$.  Let $z \in \mathbb{Z}[1, m]$ be an integer. If we sample an integer $a$ from $\mathbb{Z}[1, z]$ uniformly at random, then with probability at least $1/2$, we have that
    \[
        a \in A \qquad \textrm{and} \qquad z - a\in A.
    \]
\end{lemma}
\begin{proof}
    By the definition of density, we have that
    \[
        \frac{A[1, z]}{z} \geq \rho_m(A) \geq \frac{3}{4}.
    \]
    This implies that 
    \begin{align*}
        |\{\textrm{$(a, z - a)$ : $a \in \mathbb{Z}[1, z]$ but $a \notin A$}\}| &\leq \frac{z}{4};\\
        |\{\textrm{$(a, z - a)$ : $a \in \mathbb{Z}[1, z]$ but $z - a \notin A$}\}| &\leq \frac{z}{4}.
    \end{align*}
    Therefore,
    \[
        |\{\textrm{$(a, z-a)$ : $a\in \mathbb{Z}[1,z]$ and $a \in A$  and $z - a \in A$}\}| \geq z - \frac{z}{4} - \frac{z}{4} \geq \frac{z}{2}.
    \]
    The stated probability follows straightforwardly.
\end{proof}

Now we are ready to prove the main result of this section.

\lemconsapbydensity*
\begin{proof}
    By Lemma~\ref{lem:k-greey-sumset-solution}, in $O(n\log n)$ time, we can sort the integers in $A$ and then obtain a data structure $\mathcal{D}$ such that given any integer $z$, within $O(k\log n)$ time, $\mathcal{D}$ can check whether $z \in k\otimes A$, and returns a solution for $z \in kA$ if $z \in k\otimes A$. Our witness stores $\mathcal{D}$.  

    Given any $z \in \mathbb{Z}[0, m]$, we use the witness to find a solution for $z \in 2kA$ as follows.  If $z = 0$, we simply set $a_1 = a_2 = \cdots = a_{2k} = 0$. Assume that $z \geq 1$. We uniformly sample an integer  $a$ from $\mathbb{Z}[1, z]$ and check whether $a \in k\otimes A$ and $z - a \in k\otimes A$ both hold. We repeat this procedure until we successfully find such $a$.   Each round can be done in $O(k\log n)$ time using the data structure $\mathcal D$.  Recall that $\rho_m(A) \geq 2/k$. By Lemma~\ref{lem:constant-density} and~\ref{lem:sample-solution}, we need to sample only twice in expectation.  When $a \in k\otimes A$ and $z - a \in k\otimes A$, the data structure $\mathcal D$ also returns a solution for $a \in k A$ and a solution for $z - a \in k A$. They together form a solution for $z \in 2kA$.
\end{proof}

\section{Long Arithmetic Progressions in Sumsets}\label{sec:ka}
Lemma~\ref{lem:ap-by-density-cons} differs from Theorem~\ref{thm:ka} in two aspects.
\begin{enumerate}[label=(\roman*)]
    \item  Lemma~\ref{lem:ap-by-density-cons} requires that 
            \[
                k \geq \frac{2}{\rho_m(A)}, \quad \text{or equivalently,} \quad \rho_m(A)\geq \frac{2}{k}.
            \]
            This condition is rather strong in the sense that it requires that
            \(
                |A[1,z]| \geq \frac{2z}{k}
            \)
            for all $z \in \mathbb{Z}[1, m]$. In contrast, Theorem~\ref{thm:ka} only requires that $|A| \geq (m+1)/k$.

    \item Lemma~\ref{lem:ap-by-density-cons} additionally requires that $\{0,1\} \subseteq A$.
\end{enumerate}
We shall tackle these two issues separately. In Subsection~\ref{subsec:den-2-card}, we shall replace the density condition $\rho_m(A)\geq \frac{2}{k}$ in Lemma~\ref{lem:ap-by-density-cons} with cardinality condition $|A| \geq (m+1)/k$, and show that the same conclusion holds. In Subsection~\ref{subsec:short-ap}, we shall further remove the assumption that $\{0,1\} \subseteq A$, and show that we can still obtain an arithmetic progression in $kA$ despite that its length may not be as large as $m$. Then in Subsection~\ref{subsec:aug} we present our iterative augmentation framework which augments the arithmetic progression until its length increases to $m$.

We remark that the arguments in Subsection~\ref{subsec:den-2-card} and~\ref{subsec:short-ap} mainly follow from S{\'a}rk{\"o}zy's proof~\cite{Sar89}, except that we need to make sure that several parameters involved in the proof can be computed in near-linear time.  

\subsection{From Density to Cardinality}\label{subsec:den-2-card}
We show that when $A$ has a large cardinality, it must have a large density over some sub-interval of $[0, m]$. The following lemma is essentially the same as that in~\cite{Sar89}. We give a proof for completeness.
\begin{restatable}{lemma}{lemdensehalf}{\normalfont({\cite[Lemma 2]{Sar89}})}
    \label{lem:dense-half-exist}
    Let $A \subseteq \mathbb{Z}[0, m]$ be a set of $n$ integers. Let $k$ be a positive integer. If 
    \[
        n \geq \frac{m+1}{k},
    \]
    then there exists $u \in \mathbb{Z}[-1,m+1]$ satisfying one of the followings.
    \begin{enumerate}[label={\normalfont (\roman*)}]
        \item $-1 \leq u \leq m/2$ and for every $v \in \mathbb{Z}[u + 1, m]$,
        \[
            |A[u + 1 , v]| \geq \frac{v - u}{2k}.
        \]

        \item $m/2 < u \leq m + 1$ and for every $v \in \mathbb{Z}[0, u-1]$,
        \[
            |A[v, u - 1]| \geq \frac{u - v}{2k}.
        \]
    \end{enumerate}
\end{restatable}
\begin{proof}
    Suppose, for the sake of contradiction, that no such $u$ exists.  

    We first show that 
    \[
        |A[0, m/2]| < \frac{m+1}{2k}.
    \]
    Let $z_0 = -1$. We recursively define $z_1, z_2, \ldots, z_h$ as follows. If $z_i > m/2$, then $h = i$ and we stop. Otherwise, there must be a $v \in \mathbb{Z}[z_i+1, m]$ such that 
    \[
        |A[z_i + 1, v]| < \frac{v - z_i}{2k}.
    \]
    Let $z_{i+1} = v$. Clearly, we ends with a sequence $-1 = z_0 < z_1 < \cdots < z_h \leq m$ such that $z_h > m/2$ and that
    \[
        |A[z_i + 1, z_{i+1}| < \frac{z_{i+1} - z_i}{2k}.
    \]
    Then we have that
    \[
        |A[0, m/2]| < |A[0, z_h]| < \sum_{i = 1}^h |A[z_i+1, z_{i+1}]| < \sum_{i = 1}^h \frac{z_{i+1} - z_i}{2k} \leq \frac{m+1}{2k}.
    \]

    Similarly, we can show that
    \[
        |A[m/2, m]| < \frac{m+1}{2k}.
    \]
    But then 
    \[
        |A[0, m]| = |A[0, m/2]| + |A[m/2, m]| < \frac{m+1}{2k} + \frac{m+1}{2k} = \frac{m+1}{k}.
    \]  
    Contradiction.
\end{proof}

The above lemma only proves the existence of $u$ but does not give a way to compute it. It can be proved that $u$ can be computed in $O(n)$ time. We defer the proof of Lemma~\ref{lem:dense-half-cons} to Appendix~\ref{app:comp-u}.

\begin{restatable}{lemma}{densehalfcons}\label{lem:dense-half-cons}
    The integer $u$ in Lemma~\ref{lem:dense-half-exist} can be computed in $O(n\log n)$ time.
\end{restatable}

Lemma~\ref{lem:dense-half-exist} basically states that if $|A| \geq \frac{m+1}{k}$, then $A$ has a density of at least $\frac{1}{2k}$ over some sub-interval and the length of this subinterval is at least $m/2$. Then we can use the part of $A$ within this subinterval to generate a long arithmetic progression via Lemma~\ref{lem:ap-by-density-cons}, and obtain the following lemma.

\begin{lemma}\label{lem:cons-ap-by-cardinality-restricted}
    Let $A \subseteq \mathbb{Z}[0,m]$ be a set of $n$ integers. Assume that $\{0,1\} \subseteq A$.  Let $k$ be a positive integer. Assume that
    \[
            n \geq \frac{m+1}{k}.
    \]  
    In $O(n\log n)$ time, we can compute an arithmetic progression 
    \[
        \{s\} + \{0,1, \ldots, m\} \subseteq 32kA,
    \]
    as well as a witness with $O(k \log n)$ expected query time.
\end{lemma}
\begin{proof}
    By Lemma~\ref{lem:dense-half-cons}, in $O(n\log n)$ time, we can compute an integer $u \in \mathbb{Z}[-1, m+1]$ satisfying Lemma~\ref{lem:dense-half-exist}.

    \noindent\textbf{Case (i)}: $-1 \leq u \leq m/2$ and for every $v \in \mathbb{Z}[u + 1, m]$,
        \begin{equation}\label{eq:lem-cons-ap-by-card-1}
            |A[u + 1 , v]| \geq \frac{v - u}{2k}.
        \end{equation}
    Let $B = \{0\} \cup \{b : 1 \leq b \leq m/2, u + b \in A\}$. In view of~\eqref{eq:lem-cons-ap-by-card-1}, we have that 
    \[
        \rho_m(B) \geq \frac{1}{4k}.
    \]
    Note that $|B| \leq |A| = n$. By Lemma~\ref{lem:ap-by-density-cons}, in $O(n \log n)$ time, we can construct a witness ${\mathcal W}$ with $O(k \log n)$ expected query time for 
    \[
        \{0, 1, \ldots, m\} \subseteq 16kB.
    \]
    Let $s = 16k(u+1)$. Below we show that $\mathcal W$ can also be used as a witness for 
    \[
        \{s\} + \{0, 1, \ldots, m\} \subseteq 32kA.
    \]
    (Note that such a witness also certifies the existence of this arithmetic progression in $32kA$.)
    
    To find a solution for $s + j \in 32kA$ for $j \in \mathbb{Z}[0, m]$, we first use $\mathcal W$ to find a solution $(b_1, \ldots, b_{16k})$ for $j \in 16kB$, which takes $O(k\log n)$ time in expectation. Then we convert it to a solution for $s + j \in 32kA$ as follows. Recall that~\eqref{eq:lem-cons-ap-by-card-1} implies that $u+1 \in A$ and that the definition of $B$ implies $u + b \in A$. We claim that for each $b_i$, in $O(1)$ time, we find can two integers $a_{i,1}, a_{i,2} \in A$ with $a_{i,1} + a_{i,2} = b_i + u + 1$: if $b_i = 0$, we can set $a_{i,1} = 0$ and $a_{i_2} = u+1$, and if $b_{i} > 0$, we can set $a_{i,1} = 1$ and $a_{i_2} = u + b$. Therefore, in $O(k)$ time, we can obtain $32k$ integers from $A$ such that 
    \[
        a_{1,1} + a_{1,2} + a_{2,1} + a_{2,2} + \cdots + a_{16k,1} + a_{16k,2} = 16k(u+1) + b_1 + b_2 + \cdots + b_{16k} = s + j.
    \]
    They form a solution for $s + j \in 32kA$.

    \noindent\textbf{Case (ii)}: $m/2 < u \leq m + 1$ and for every $v \in \mathbb{Z}[0, u-1]$,
    \[
            |A[v, u - 1]| \geq \frac{u - v}{2k}.
    \]
    Let $B = \{0\} \cup \{b : 1 \leq b \leq m/2, u - b \in A\}$.  Similar to Case (i), in $O(n \log n)$ time, we can construct a witness ${\mathcal W}$ with $O(k\log n)$ expected query time for 
    \[
        \{0, 1, \ldots, m\} \subseteq 16kB.
    \]
    Also similar to Case (i), $\mathcal{W}$ can be used as a witness for 
    \[
        \{s\} + \{0, 1, \ldots, m\} \subseteq 32kA,
    \] 
    where $s = 16ku - m$.  We omit the details.
\end{proof}

\subsection{Arithmetic Progressions for General Sets}\label{subsec:short-ap}
Lemma~\ref{lem:cons-ap-by-cardinality-restricted} is very close to Theorem~\ref{thm:ka} except for that it additionally requires that $\{0,1\} \subseteq A$. We shall remove this assumption and show that we can still obtain an arithmetic progression despite that its length may not be as large as $m$.

\begin{lemma}\label{lem:cons-ap-by-cardinality-short}
    Let $A \subseteq \mathbb{Z}[0,m]$ be a set of $n$ integers. Assume that $n \geq 2$.  Let $k$ be a positive integer. Assume that
    \[
            n \geq \frac{m+1}{k}.
    \]  
    In $O(n\log n)$ time, we can compute an arithmetic progression 
    \[
        \{s\} + \{0, d, 2d \ldots, \ell d\} \subseteq 320kA
    \]
    with $d \leq \frac{2m}{n}$ and $\ell \cdot \min\{d, n\} \geq 5m$. In the same time, we can obtain a witness with $O(k \log n)$ expected query time.
\end{lemma}
\begin{proof}
    Assume that $A$ is sorted. This is without loss of generality since sorting takes only $O(n\log n)$ time. In $O(n)$ time, we can find the closest pair of integers in $A$. Suppose that they are $\{a^*, a^* + g\}$. We have that 
    \begin{equation}\label{eq:short-ap-1}
        1 \leq g \leq \frac{m}{n-1} \leq \frac{2m}{n}.
    \end{equation}
    Let $A^* = \{a^*, a^* + g\}$. Let $t = |A \mod g|$.  Let $A'$ be the largest subsets of $A$ so that the elements in $A'$ are all congruent modulo $g$. Note that $A'$ can be obtained in $O(n)$ time and that
    \[
        |A'| \geq \frac{n}{t}.
    \]
    Let $a'$ be the minimum element of $A'$. Now consider $A' + A^*$. All the elements in $A' + A^*$ are congruent to $(a'+ a^*)$ modulo $g$. In $O(n)$ time, we can compute the following set. 
    \[
        B = \{b : bg + a' + a^* \in A' + A^*\}.
    \]
    It is easy to see that $\{0,1\} \subseteq B$ and that 
    \begin{equation}\label{eq:cons-ap-by-cardinality-short-1}
        n + 1 \geq |B| \geq |A'| + 1 \geq \frac{n}{t} + 1\geq \frac{m+1}{tk} + 1 \geq \frac{1}{5k}\left(\frac{5(m+1)}{t} + 5k\right) \geq \frac{1}{5k}\left(\left\lceil\frac{5m}{t}\right\rceil + 1\right).
    \end{equation}
    We also have that 
    \begin{equation}\label{eq:cons-ap-by-cardinality-short-2}
        \max(B) \leq \frac{m  + g - a'}{g} \leq \frac{m}{t} + 1 \leq \left\lceil\frac{5m}{t}\right\rceil.
    \end{equation}
    The second inequality is due to that $t = |A \mod g| \leq g$.
    In view of~\eqref{eq:cons-ap-by-cardinality-short-1} and~\eqref{eq:cons-ap-by-cardinality-short-2} and by Lemma~\ref{lem:cons-ap-by-cardinality-restricted}, in $O(n \log n)$ time, we can compute an arithmetic progression
    \[
        \{s'\} + \left\{0, 1, \ldots, \left\lceil\frac{5m}{t}\right\rceil\right\} \subseteq 160kB,
    \]
    and construct a witness $\mathcal W$ with $O(k\log n)$ expected time for it.

    Recall that $a'$ is the minimum element in $A'$. Let $s = s'g + 160k(a' + a^*)$.   We claim that  
    \[
        \{s\} + \{0, g, 2g, \ldots, \left\lceil \frac{5m}{t}\right\rceil g\} \subseteq 320kA,
    \]
    and that $\mathcal W$ can be used as a witness for it. To prove the claim, it suffices to show that given any $j$ with $0\leq j\leq \lceil\frac{5m}{t}\rceil$, we can use $\mathcal W$ to find a solution for $s + jg \in 320kA$.

    Fix an arbitrary $j$ with $0\leq j\leq \lceil\frac{5m}{t}\rceil$. We first use $\mathcal W$ to find a solution $(b_1, b_2, \ldots, b_{160k})$ for $s' + j \in 160kB$, which takes $O(k \log n)$ time in expectation.  Next we shall convert it to a solution for $s + jg \in 320kA$ by replacing each $b_i$ with two integers $a_{i,1}, a_{i,2} \in A$ such that $bg + a' + a^* = a_{i,1} + a_{i,2}$.  Note that for each $b_i$, by definition of $B$, either $b_ig + a' + a^* = a + a^*$ or $b_ig + a' + a^* = a + a^* + g$ for some $a \in A$.  In other words, either $b_ig + a' \in A$ or $b_ig + a' - g \in A$.  These two cases can be distinguished in $O(\log n)$ time by checking whether $b_ig + a' \in A$.  In the former case, we can set $a_{i,1} = b_ig + a'$ and $a_{i,2} =  a^*$, while in the latter case, we can set $a_{i,1} = b_ig + a' - g$ and $a_{i,2} =  a^* + g$.
    As a result, in $O(k \log n)$ time, we can replace all the $160k$ integers $(b_1, b_2, \ldots, b_{160k})$ with $320$k integers $(a_{1,1}, a_{1,2}, \ldots, a_{160k,1}, a_{160k,2})$ such that 
    \[
        a_{1,1} + a_{1,2} + \cdots + a_{160k,1} + a_{160k,2} = \sum_{i=1}^{160k} (b_ig + a' + a^*) = (s' + j)g + 160k(a' + a^*) = s + jg.
    \]
    They form a solution for $s + jg \in 320kA$.

    In view of~\eqref{eq:short-ap-1}, the resulting arithmetic progression has common difference $d = g \leq \frac{2m}{n}$, and its length 
    \[
        \ell = \left\lceil\frac{5m}{t}\right\rceil \geq \frac{5m}{\min\{d, n\}}.
    \]
    The last inequality is due to that $t = |A \mod g| \leq \min\{n, g\}$.
\end{proof}

\subsection{Augmenting Arithmetic Progressions}\label{subsec:aug}
Lemma~\ref{lem:cons-ap-by-cardinality-short} gives an arithmetic progression $P = \{s\} + \{0, d, \ldots, \ell d\}$ with $ \ell \cdot \min\{n, d\} \geq 5m$. When $d  = 1$, then $\ell \geq 5m$ and we already have an arithmetic progression of length at least $m$. When $d$ is large, $\ell$ can be much smaller than $m$, and we have to augment $P$ in this case. In this section, we always assume that $0 \in A$ and $\gcd(A) = 1$.

Ideally, if for some proper divisor $d'$ of $d$, we have another set $Q = \{s_q\} + \{0, d', 2d', 3d', \ldots, (\frac{d}{d'}-1)\cdot d'\}$, then one can verify that 
\[
    \{s + s_q\} + \{0, d', 2d', \ldots, \frac{\ell d}{d'} \cdot d'\} \subseteq P + Q.
\]
That is, we can augment $P$ using $Q$ so that its length increases by a factor of $\frac{d}{d'}$ and its common difference decreases by the same factor. Such a set $Q$ is, however, too good to be generated using $A$. We shall use a less restricted set $Q$. Roughly speaking, in the following lemma, we use a set $Q$ that contains $\{0, d', 2d', \ldots, (\frac{d}{d'}-1)\cdot d'\}$ when modulo $d$.

\begin{lemma}\label{lem:aug-idea-1}
    Let $P = \{s\} + \{0, d, 2d, \ldots, \ell d\}$. Let $d'$ be a proper divisor of $d$.  Let $Q$ be a set that contains $\frac{d}{d'}$ integers $\{q_0, \ldots, q_\frac{d-d'}{d'}\}$ satisfying the following property: there exists an integer $s_q$ such that for every $i = 0, \ldots, \frac{d}{d'}-1$,
    \[
        q_i \equiv s_q + id' \pmod d.
    \]
    Let $h_{\max} \geq \max_i\lfloor\frac{q_{i} - s_q}{d}\rfloor$ and $h_{\min} \leq \min_i \lfloor\frac{q_i - s_q}{d}\rfloor$ be two integers. Then there exists some arithmetic progression
    \[
        \{s'\} + \{0, d', 2d', \ldots, \ell' d'\} \subseteq P + Q,
    \]
    where $s' = s + s_q + h_{\max}d$ and $\ell' = (\ell - h_{\max} + h_{\min} )\cdot \frac{d}{d'}$. Moreover, for any term $s' + j'd'$ in the resulting arithmetic progression, there is a solution $(p, q_i) \in P \times Q$ for $s' +j'd'\in P+Q$ satisfying
    \[
        i = \frac {j'd' \bmod d}{d'}.
    \]
\end{lemma}
\begin{proof}
    For each $q_i$, let $k_i = (q_i - s_q - id')/d$. In other words, $q_i = s_q + id' + k_id$.  We have that 
    \begin{equation}\label{eq:aug-idea-1-1}
        h_{\min}-1 \leq \frac{q_i - s_q -d}{d}  < k_i \leq \frac{q_i - s_q}{d} \leq h_{\max}.
    \end{equation}
    To prove the lemma, it suffices to show that for any $j'\in \mathbb{Z}[0, \ell']$, we can find a solution for $s' + j'd' \in P + Q$. Fix an arbitrary $j' \in \mathbb{Z}[0, \ell']$.  Let $i = (j'd' \bmod d)/d'$. It is easy to see that $i \in \mathbb{Z}[0, \frac{d-d'}{d'}]$. Therefore, $q_i$ is well-defined. It remains to show that $s' + j'd' - q_i \in P$. Plug in $s'$ and $q_i$, we have
    \[
        s' + j'd' - q_i = s + s_q +  h_{\max}d + j'd' - s_q - id' - k_id = s + (h_{\max} - k_i)d + (j' - i)d'.
    \]
    By the definition of $i$, we have that $d | (j' - i)d'$ and that 
    \[
        0 \leq (j' - i)d' \leq j'd' \leq \ell'd' = (\ell - h_{\max} + h_{\min})d.
    \]
    This implies that $d | (h_{\max} - k_i)d + (j' - i)d'$ and that 
    \[
        0 \leq (h_{\max} - k_i)d + (j' - i)d' \leq (h_{\max} - k_i)d + (\ell - h_{\max} + h_{\min})d = (\ell + h_{\min} - k_i)d \leq \ell d.
    \]
    The first and the last inequalities are due to~\eqref{eq:aug-idea-1-1}.  Therefore, $s' + j'd' - q_i = s + jd$ for some $j \in \mathbb{Z}[0, \ell]$, which implies that $s' + j'd' - q_i \in P$.  The above also implies that there is a solution $(p, q_i) \in P \times Q$ for $s' + j'd' \in P + Q$ satisfying 
    $i = (j'd' \bmod d)/d'$.
\end{proof}

\subsubsection{Augmenting Using Pairs "Not Divisible" by d}
We show that the set $Q$ in Lemma~\ref{lem:aug-idea-1} is obtainable whenever we have a pair of integers $\{a^*, a^* + g\}$ with $d \nmid g$. Therefore, we can augment $P$ when we have such a pair.

\begin{lemma}\label{lem:remainder-ka}
    Let $d$ be an integer with $d \geq 2$. Let $A = \{a, a + g\}$ be a set of integers with $d \nmid g$. In $O(\log \frac{d}{d'})$ time, we can compute a proper divisor of $d$ such that for $i = 0, 1, \ldots, \frac{d}{d'}-1$, there is an integer $q_i \in \frac{d}{d'}A$ such that 
    \[
        q_i \equiv \frac{da}{d'} + id' \pmod d.
    \]
    Moreover, given any $i \in \mathbb{Z}[0, \frac{d}{d'} - 1]$, in $O(\frac{d}{d'})$ time, we can compute $q_i$ and a solution for $q_i \in (\frac{d}{d'} - 1)A$.
\end{lemma}
\begin{proof}
    Using the Euclidean algorithm~\cite{shallit1994origins}, in $O(\log \frac{d}{d'})$ time, we can obtain $d'=\gcd(d, g)$ and an integer $j^* \in \mathbb{Z}[0, \frac{d - d'}{d'}]$ with 
    \[
        j^* g \equiv d' \pmod d.
    \]

    Now consider $\frac{d}{d'}A$. It is easy to see that
    \[
        \frac{d}{d'}A = \{\frac{da}{d'}\} + \{0, g, 2g, \ldots, \frac{d}{d'}g\}. 
    \]
    
    Given any $i$ with $0\leq i \leq \frac{d}{d'}-1$, in $O(1)$ time, we can find an integer $q_i \in \frac{d}{d'}A$ with $q_i \equiv \frac{da}{d'} + id' \pmod d$ as follows. Recall that we have an integer $j^* \in \mathbb{Z}[0, \frac{d - d'}{d'}]$ with 
    \[
        j^* g \equiv d' \pmod d.
    \]  
    Let $j = ij^* \bmod \frac{d}{d'}$. Clearly, $0 \leq j \leq \frac{d}{d'} - 1$. It remains to show that $jg \equiv id' \pmod d$.
    By some arithmetic calculation, one can verify that for some integers $k$ and $h$, 
    \[
        jg = id' + ikd - \frac{hgd}{d'}
    \]
    Note that $d' \mid g$ as $d' = \gcd(d,g)$. Therefore, $d \mid (ikd - \frac{hgd}{d'})$, which implies that
    \begin{equation}\label{eq:aug-method-1-1}
        jg \equiv  i d' \pmod d.
    \end{equation}
    Now let $q_i = \frac{da}{d'} + jg$. We have that 
    \[
        q_i \equiv \frac{da}{d'} + id' \pmod d.
    \]

    Note that 
     \[
        q_i = \frac{da}{d'} + jg  = (\frac{d}{d'} - j)a + j(a+g).
    \]
    A solution for $q_i \in \frac{d}{d'}A$ can be constructive straightforwardly in $O(\frac{d}{d'})$ time.
\end{proof}
We remark that one can actually use $(\frac{d}{d'} - 1)A$ in the above lemma. We use $\frac{d}{d'}A$ for simplicity, and asymptotically it makes no difference in the final result.

Combing Lemma~\ref{lem:aug-idea-1} and Lemma~\ref{lem:remainder-ka}, we have the following.
\begin{lemma}\label{lem:aug-method-1}
    Let $P = \{s\} + \{0, d, 2d, \ldots, \ell d\}$ with $d > 1$.  Let $A = \{a, a + g\}$ for some non-negative integers $a$ and $g$ with $d \nmid g$.  In $O(\log \frac{d}{d'})$ time, we can compute a proper divisor $d'$ of $d$ and an arithmetic progression
    \[
        \{s'\} + \{0, d', 2d', \ldots, \ell'd'\} \subseteq P + \frac{d}{d'}A,
    \]
    where $\ell' \geq (\ell - \frac{g}{d'})\cdot \frac{d}{d'}$.  In the same time, we can build a witness with $O(\frac{d}{d'})$ query time.
\end{lemma}
\begin{proof}
    Let $d'$ be the proper divisor of $d$ given by Lemma~\ref{lem:remainder-ka}. Then $\frac{d}{d'}A$ can be used as the set $Q$ to augment $P$ via Lemma~\ref{lem:aug-idea-1}. Note that $s_q = \frac{da}{d'}$, $h_{\min} = 0$, and $h_{\max} = \frac{g}{d'}$. Therefore, 
    \[
        \{s + \frac{da}{d'} + \frac{dg}{d'}\} + \{0, d', 2d', \ldots, \ell'd'\} \subseteq P + \frac{d}{d'}A
    \]
    with $\ell' \geq (\ell - \frac{g}{d'})\frac{d}{d'}$.

    Let $s'  = s + \frac{da}{d'} + \frac{dg}{d'}$. To find a solution for $s' + j'd' \in P + \frac{d}{d'}A$, we first use Lemma~\ref{lem:aug-idea-1} to find the index $i$ of the corresponding $q_i$, then we use this index $i$ to find $q_i$ and a solution for $q_i \in \frac{d}{d'}A$ via Lemma~\ref{lem:remainder-ka}. Then let $p = s' + j'd' - q_i$. Lemma~\ref{lem:aug-idea-1} implies that $p \in P$. The solution for  $q_i \in \frac{d}{d'}A$, together with $p$, form a solution for $s' + j'd' \in P + \frac{d}{d'}A$. One can verify that the query time is $O(\frac{d}{d'})$.
\end{proof}

\subsubsection{Augmenting Using Pairs "Divisible" by d}
Lemma~\ref{lem:aug-method-1} increase the length of $P$ from $\ell$ to $(\ell - \frac{g}{d'})\cdot \frac{d}{d'}$. When $g$ is small, it increases the length of $P$ by a large factor. When $g$ is large, however, it may not be as effective as we need.  In this case, we need another method to make up for the loss of length due to large $g$. The second method augments $P$ using a pair from $A$ whose gap is a multiple of $d$.

\begin{lemma}\label{lem:aug-method-2}
    Let $P = \{s\} + \{0, d, 2d, \ldots, \ell d\}$.  Let $A = \{a, a + g\}$ for some non-negative integers $a$ and $g$. Assume that $d \mid g$ and that $1\leq g \leq \ell d$. Then for any $h > 0$, in $O(1)$ time, we can compute an arithmetic progression
    \[
        \{s + ha\} + \{0, d, 2d, \ldots, (\ell + \frac{hg}{d})\cdot d\} \subseteq P + hA.
    \]
    as well as a witness with $O(h)$ query time for it.
\end{lemma}
\begin{proof}
    Observe that 
    \[
        hA = \{ha\} + \{0, g, \ldots, hg\}.
    \]
    Let $s' = s + ha$. To prove the lemma, we show that for any $j' \in \mathbb{Z}[0, \ell + \frac{hg}{d}]$, we can find a solution $(p, a_1, \ldots, a_h)$ for $s' + j'd \in P + hA$. When $\frac{hg}{d} \leq j' \leq \ell + \frac{hg}{d}$, this is trivial since we can simply let $p = s + (j' - \frac{hg}{d})d$ and $a_1 = a_2 =\cdots = a_h = a + g$. Assume that $j' < \frac{hg}{d}$.
     Let $q = \lfloor j'd / g \rfloor$, and let $j = (j'd \mod g)/d$.  Then 
    \[
        s' + j'd = s + ha + j'd = s + jd + ha + qg.
    \]
    We shall show that $s + jd \in P$ and $ha + qg \in hA$.  Recall that $g \leq \ell d$. By definition of $j$, we have that $j \leq \frac{g}{d} \leq \ell$. Therefore, $s + jd \in P$. Since $j' < \frac{hg}{d}$, we have that $q < h$, so $ha + qg \in hA$. Let $p = s + jd$, let $a_1 = \cdots = a_q = a + g$, and let $a_{q+1} = \cdots = a_{h} = a$. $(p, a_1, \ldots, a_h)$ forms a solution for $s' + j'd \in P + hA$.

    Let the witness store $d$ and $g$. It is easy to see that the above procedure to find $(p, a_1, \ldots, a_h)$ takes only $O(h)$ time.
\end{proof}

\subsubsection{Combining the Two Approaches of Augmenting}

One can see that, in contrast to the first method, the second method prefers a pair with a large gap $g$ (as long as $d\mid g$ and $g \leq \ell d$). A natural approach is to use these two methods simultaneously.  A crucial observation is that $A$ either contains a pair $(a, a+g)$ of integers with $g$ small and not divisible by $d$, or a pair $(a', a' + g)$ with $g$ large and divisible by $d$.  In the former case, the first method has a good effect; in the latter case, although the first method has a big loss in the length of the resulting arithmetic progression, this loss can be made up by the second method.

\begin{lemma}\label{lem:small-and-large-gap}
    Let $A \subseteq \mathbb{Z}[0, m]$ be a sorted set of $n$ integers. Assume that $0 \in A$ and that $\gcd(A) = 1$.  Let $d > 1$ be an integer. In $O(n)$ time, we can find either
    \begin{itemize}
        \item a pair $(a, a + g) \in A \times A$ such that $d \nmid g$ and that $1\leq g \leq \frac{4m}{n}$, or

        \item two pairs $(a, a + g)$ and $(a', a' + g')$ in $A \times A$ such that $d \nmid g$, $d \mid g'$, and that
        \[
            \frac{nd}{4m}g \leq g' \leq m.
        \]
    \end{itemize}
\end{lemma}
\begin{proof}
    Let $a_1 < a_2 < \ldots < a_n$ be the elements of $A$ labeled in increasing order. For $i \in \mathbb{Z}[1, n-1]$, let $g_i = a_{i+1} - a_i$. Let $h$ be the number of $g_i$'s that are not divisible by $d$. Since $\gcd(A) = 1$ and $0 \in A$, we have that $h \geq 1$.  Let $g_{i^*}$ be the smallest $g_i$ that are not divisible by $d$. It can be found in $O(n)$ time.

    Note that $g_i > 0$ and that $\sum_{i} g_i \leq m$. Therefore, $g_{i^*} \leq m/h$. If $h \geq n/4$, then we are done since
    \[
        g_{i^*} \leq \frac{m}{h} \leq \frac{4m}{n}.
    \]

    Suppose that $h < n/4$. Let $i_1, \ldots, i_h$ be the indices of the $g_i$'s that are not divisible by $d$. These $h$ gaps separate the other gaps into at most $h+1$ groups. For simplicity define $i_0 = -1$ and $i_{h+1} = n+1$. For $j \in \mathbb{Z}[0, h]$, define 
    \[
        G_j = \{g_i : i_j < i < i_{j+1}\}
    \]
    to be the set of $g_i$'s whose indices are between $i_j$ and $i_{j+1}$. Since $\sum_j |G_j| = n - 1 - h$, there must be some $G_{j^*}$ with 
    \[
        |G_{j^*}| \geq \frac{n-1-h}{h+1} \geq \frac{n}{h+1} - 1 \geq \frac{n}{2h} - 1 \geq \frac{n}{4h}.
    \]
    The last inequality is due to that $h < n / 4$. Such a $G_{j^*}$ can be found in $O(n)$ time. Moreover, every $g_i \in G_{j^*}$  is a positive integer divisible by $d$, and hence is at least $d$. Therefore, we have that $d \mid \sum_{g_i\in G_{j^*}} g_i$ and that 
    \[
        \sum_{g_i\in G_{j^*}} g_i \geq d|G_{j^*}| \geq \frac{nd}{4h}.
    \]
    Recall that $g_{i^*} \leq \frac{m}{h}$. Therefore,
    \[
        \sum_{g_i\in G_{j^*}} g_i \geq \frac{nd}{4m}g_{i^*}.
    \]
    Note that $\sum_{g_i\in G_{j^*}} g_i = a_{i_{j^*+1}} - a_{i_{j^*}+1}$ and that $g_{i^*} = a_{i^*+1} - a_{i^*}$.  Taking $g'=\sum_{g_i\in G_{j^*}} g_i$ and $g=g_{i^*}$ finishes the proof.
\end{proof}

As a result of Lemma~\ref{lem:aug-method-1}, Lemma~\ref{lem:aug-method-2}, and Lemma~\ref{lem:small-and-large-gap}, we can always augment $P$ so its length increases by a large amount.

\begin{lemma}\label{lem:augment-method-ult}
    Let $A \subseteq \mathbb{Z}[0, m]$ be a sorted set of $n$ integers. Assume that $0 \in A$ and that $\gcd(A) = 1$. Let $P = \{s\} + \{0, d, 2d, \ldots, \ell d\}$ with $\ell d \geq m$ and $d > 2$. In $O(n + \log \frac{d}{d'})$ time, we can compute a proper divisor $d'$ of $d$ and also an arithmetic progression 
    \[
        \{s'\} + \{0, d', 2d', \ldots, \ell'\cdot d'\} \subseteq P + (\frac{d}{d'} + \left\lceil\frac{4m}{nd'}\right\rceil)A.
    \]
    where $\ell' \geq (\ell - \frac{4m}{nd'})\frac{d}{d'}$. In the same time, we can build a witness with $O(\frac{d}{d'} + \left\lceil\frac{4m}{nd'}\right\rceil)$ query time.
\end{lemma}
\begin{proof}
    Consider the two cases of Lemma~\ref{lem:small-and-large-gap}. 

    \noindent\textbf{Case (i).} In $O(n)$ time, we find a pair $(a, a+g) \in A\times A$ with $d\nmid g$ and $1\leq g \leq \frac{4m}{n}$.  Let $A^* = \{a, a + g\}$. Let $d' = \gcd(d, g)$, which can be computed in $O(\log \frac{d}{d'})$ time.  By Lemma~\ref{lem:aug-method-1}, in $O(\log d)$ time, we can compute an divisor $d'$ of $d$ with $d' < d$ and also an arithmetic progression
    \[
        \{s'\} + \{0, d', 2d', \ldots, \ell'd'\} \subseteq P + \frac{d}{d'}A^*,
    \]
    where $\ell' \geq (\ell - \frac{g}{d'})\cdot \frac{d}{d'} \geq (\ell - \frac{4m}{nd'})\cdot \frac{d}{d'}$.   In the same time, we can also build a witness $\mathcal W$ with $O(\frac{d}{d'})$ query time.  Since $A^* \subseteq A$ and $0 \in A$, any solution for $z \in P + \frac{d}{d'}A^*$ can be converted to a solution for $z \in P + (\frac{d}{d'} + \left\lceil\frac{4m}{nd'}\right\rceil)A$ by appending $0$'s.  Therefore, 
    \[
        \{s'\} + \{0, d', 2d', \ldots, \ell'd'\} \subseteq P + (\frac{d}{d'} + \left\lceil\frac{4m}{nd'}\right\rceil)A,
    \]
     and $\mathcal W$ can be used as a witness with $O(\frac{d}{d'} + \left\lceil\frac{4m}{nd'}\right\rceil)$ query time for this progression.

    \noindent\textbf{Case (ii).} In $O(n)$ time, we find two pairs $(a, a + g)$ and $(a', a' + g')$ in $A \times A$ such that $d \nmid g$, $d \mid g'$, and that
        \begin{equation}\label{eq:augment-method-ult-1}
            \frac{nd}{4m}g \leq g' \leq m.
        \end{equation}
    Let $A^*_1 = \{a, a + g\}$ and let $A^*_2 = \{a', a' + g'\}$. Let $d' = \gcd(d, g)$, which can be computed in $O(\log \frac{d}{d'})$ time. Let $h = \left\lceil\frac{4m}{nd'}\right\rceil$. We first use $A^*_2$ to augment $P$. Since $\ell d \ge m\ge g'$, by Lemma~\ref{lem:aug-method-2}, in $O(1)$ time, we can compute an arithmetic progression 
    \[
        P'' = \{s''\} + \{0, d, 2d, \ldots, \ell''\cdot d\} \subseteq P + hA^*_2
    \]
    and a witness $\mathcal{W}$ with $O(h)$ query time, where 
    \begin{equation}\label{eq:augment-method-ult-2}
        \ell'' = \ell + \frac{hg'}{d} = \ell + \left\lceil\frac{4m}{nd'}\right\rceil \cdot \frac{g'}{d} \geq \ell + \frac{g}{d'}.
    \end{equation}
    The last inequality is due to~\eqref{eq:augment-method-ult-1}.

    Then we use $A^*_1$ to augment $P''$.  By lemma~\ref{lem:aug-method-1}, in $O(\log \frac{d}{d'})$ time, we can compute an arithmetic progression 
    \[
        P' = \{s'\} + \{0, d', 2d', \ldots, \ell'd'\} \subseteq P'' + \frac{d}{d'}A^*_1
    \]
    with $\ell' \geq (\ell'' - \frac{g}{d'})\cdot \frac{d}{d'}$ and a witness $\mathcal{W'}$ with $O(\frac{d}{d'})$ query time.  In view of~\eqref{eq:augment-method-ult-2}, we have that
    \[
        \ell' \geq (\ell'' - \frac{g}{d'})\frac{d}{d'} \geq \ell \cdot \frac{d}{d'} 
    \]

    It is easy to see that 
    \[
        P' \subseteq P'' + \frac{d}{d'}A^*_1 \subseteq P + hA^*_2 + \frac{d}{d'}A^*_1 \subseteq P + (\frac{d}{d'} + \left\lceil\frac{4m}{nd'}\right\rceil)A.
    \]
     and that $\mathcal{W}$ and $\mathcal{W'}$ together serve a witness with $O(\frac{d}{d'} + \left\lceil\frac{4m}{nd'}\right\rceil)$ query time for 
     \[
        P' \subseteq P + (\frac{d}{d'} + \left\lceil\frac{4m}{nd'}\right\rceil)A. \qedhere
    \]
\end{proof}

\subsubsection{Augmenting Iteratively}
When $P$ is not too short, we can augment $P$ by repeatedly invoking Lemma~\ref{lem:augment-method-ult} until its length is at least $m$.
\begin{lemma}\label{lem:augment-iterative}
    Let $A \subseteq \mathbb{Z}[0, m]$ be a sorted set of $n$ integers. Assume that $0 \in A$ and that\\$\gcd(A) = 1$. Let $P = \{s\} + \{0, d, 2d, \ldots, \ell d\}$ with 
    \begin{equation}\label{eq:augment-iterative-0}
        \ell \cdot \min\{d,n\} \geq 5m.
    \end{equation}
    In $O(n \log d)$ time, we can compute an arithmetic progression 
    \[
        \{s'\} + \{0, 1, \ldots, m\} \subseteq P + \left(2d + \left\lceil\frac{8m}{n}\right\rceil\right)A.
    \]
    and a witness with $O(d + \frac{m}{n})$ query time.
\end{lemma}
\begin{proof}
    We first define a series of arithmetic progressions $P_0, P_1, \ldots, P_h$ as follows. 
    We denote $P$ as $P_0 = \{s_0\} + \{0, d_0, \ldots, \ell_0d_0\}$. Given $P_i$, if $d_i = 1$, then $h = i$; if $d_i \geq 2$, let $P_{i+1} = \{s_{i+1}\} + \{0, d_{i+1}, \ldots, \ell_{i+1}d_{i+1}\}$ be the arithmetic progression obtained from $P_{i}$ by applying Lemma~\ref{lem:augment-method-ult}.  Clearly, $h \leq \log d_0$ as $d_{i+1}$ is a proper divisor of $d_i$. 

    We first show that for any $i$, we have that $\ell_id_i \geq m$. It will certify that every $P_i$ (except for $P_h$) satisfies the precondition (i.e., $\ell d \geq m$) of Lemma~\ref{lem:augment-method-ult} and therefore, Lemma~\ref{lem:augment-method-ult} can be applied to generate $P_{i+1}$.  Moreover, it will imply that $\ell_h \geq m$ as $d_h = 1$. By Lemma~\ref{lem:augment-method-ult},
    \[
        \ell_{i + 1} \geq (\ell_i - \frac{4m}{nd_{i+1}})\frac{d_i}{d_{i+1}},
    \]
    or equivalently,
    \[
        \ell_{i+1}d_{i+1} \geq \ell_id_i - \frac{4m}{n}\cdot \frac{d_{i}}{d_{i+1}}.
    \]
    This implies that 
    \begin{equation}\label{eq:augment-iterative-1}
        \ell_id_i \geq \ell_0d_0 - \frac{4m}{n} \sum_{j = 0}^{i-1}\frac{d_j}{d_{j+1}} \geq \ell_0d_0 - \frac{4m}{n}\cdot d_0.
    \end{equation}
    The second inequality is due to that $d_{j+1}$ is a proper divisor of $d_{j}$. Then in view of~\eqref{eq:augment-iterative-0} and~\eqref{eq:augment-iterative-1}, we have that 
    \[
        \ell_id_i \geq (\ell_0 - \frac{4m}{n})d_0 \geq \frac{\ell_0d_0}{5} \geq m.
    \]

    Then we show that 
    \[
        P_h \subseteq P_0 + \left(2d_0 + \left\lceil\frac{8m}{n}\right\rceil\right)A
    \]
    By Lemma~\ref{lem:augment-method-ult}, 
    \[
        P_{i+1} \subseteq P_i + (\frac{d_i}{d_{i+1}} + \lceil \frac{4m}{nd_{i+1}} \rceil)A.
    \]
    Therefore
    \begin{equation}\label{eq:augment-iterative-2}
        P_h \subseteq P_0 + (\sum_{i=0}^{h-1}\frac{d_i}{d_{i+1}} +\sum_{i=0}^{h-1}\lceil \frac{4m}{nd_{i+1}} \rceil)A.
    \end{equation}
    Since $d_{i+1}$ is a proper divisor of $d_i$ with $d_{i+1} < d_i$, we have that 
    \begin{equation}\label{eq:augment-iterative-3}
        \sum_{i=0}^{h-1}\frac{d_i}{d_{i+1}} \leq d_0
    \end{equation}
    and that
    \begin{align}
        \sum_{i=0}^{h-1}\lceil \frac{4m}{nd_{i+1}} \rceil &\leq \sum_{i=0}^{h-1}(\frac{4m}{nd_{i+1}} + 1)\nonumber\\
            & \leq h + \frac{4m}{n}(\frac{1}{d_1} + \ldots + \frac{1}{d_h}) \nonumber\\
            &\leq \log d_0 + \frac{4m}{n}(\frac{1}{2^{h-1}} + \frac{1}{2^{h-2}} + \cdot + \frac{1}{2^0}) \leq d_0 + \lceil\frac{8m}{n}\rceil\label{eq:augment-iterative-4}
    \end{align}
    Recall that $0 \in A$. In view of~\eqref{eq:augment-iterative-2},~\eqref{eq:augment-iterative-3}, and~\eqref{eq:augment-iterative-4}, we have 
    \[
        P_h \subseteq P_0 + \left(2d_0 + \left\lceil\frac{8m}{n}\right\rceil\right)A.
    \]

    The time needed to compute $P_{i+1}$ from $P_{i}$ is $O(n + \log \frac{d_{i}}{d_{i+1}})$. Recall that $h \leq \log d_0$. Therefore, the total time to compute $P_h$ is 
    \[
        O(n \log d_0 + \sum_{i=0}^h \log \frac{d_i}{d_{i+1}}) = O(n\log d_0).
    \]

    As to the witness, Lemma~\ref{lem:augment-method-ult} also provides a witness $\mathcal{W}_{i+1}$ with $O(\frac{d_i}{d_{i+1}} + \lceil \frac{4m}{nd_{i+1}} \rceil)$ query time for $P_{i+1} \subseteq P_i + (\frac{d_i}{d_{i+1}} + \lceil \frac{4m}{nd_{i+1}} \rceil)A$. All these witnesses can be combined to form a witness for 
    \[
         P_h \subseteq P_0 + \left(2d_0 + \left\lceil\frac{8m}{n}\right\rceil\right)A.
    \]
    The query time is $O(d_0 + \frac{m}{n})$.
\end{proof}

\subsection{Putting Things Together}
Now we can show that $kA$ contains a long arithmetic progression by first generating an arithmetic progression of moderate length using Lemma~\ref{lem:cons-ap-by-cardinality-short} and then augmenting it using Lemma~\ref{lem:augment-iterative}.
\begin{lemma}\label{lem:cons-ap-by-cardinality-long}
    Let $A \subseteq \mathbb{Z}[0,m]$ be a set of $n$ integers. Assume that $0 \in A$ and that $\gcd(A) = 1$.  Let $k$ be a positive integer. Assume that
    \[
            n \geq \frac{m + 1}{k}.
    \]  
    In $O(n\log m)$ time, we can compute an arithmetic progression 
    \[
        \{s\} + \{0, 1, 2, \ldots, m\} \subseteq 332kA,
    \]
    and a witness with $O(k \log n)$ expected query time.
\end{lemma}
\begin{proof}
    By Lemma~\ref{lem:cons-ap-by-cardinality-short},  in $O(n\log n)$ time, we can compute an arithmetic progression 
    \begin{equation}\label{eq:cons-ap-by-cardinality-long-1}
        P = \{s\} + \{0, d, 2d \ldots, \ell d\} \subseteq 320kA
    \end{equation}
    with $d \leq \frac{2m}{n}$ and $\ell\cdot \min\{d, n\} \geq 5m$,
    as well as a witness $\mathcal W$ with $O(k \log n)$ expected query time.  

    Then we augment $P$ via Lemma~\ref{lem:augment-iterative}. In $O(n \log d)$ time, we can compute an arithmetic progression
    \begin{equation}\label{eq:cons-ap-by-cardinality-long-2}
        P' = \{s'\} + \{0,1, \ldots, m\} \subseteq P + (2d + \lceil\frac{8m}{n}\rceil)A.
    \end{equation}
    and a witness $\mathcal{W}'$ with $O(d + \frac{m}{n})$ query time.

    In view of~\eqref{eq:cons-ap-by-cardinality-long-1} and~\eqref{eq:cons-ap-by-cardinality-long-2}, we have that 
    \[
        P' \subseteq (320k + 2d + \lceil\frac{8m}{n}\rceil)A.
    \]
    Moreover, $\mathcal{W}$ and $\mathcal{W}'$ together form a witness $\mathcal{W}''$ with $O(k \log n + d + \frac{m}{n})$ expected query time.

    Note that
    \[
        \frac{m}{n} \leq k
    \]
    and that 
    \[
        d \leq \frac{2m}{n} \leq 2k.
    \]
    As $0 \in A$, we have that 
    \[
        P' \subseteq (320k + 4k + 8k)A = 332kA.
    \]
    $\mathcal{W''}$ can also be used as a witness for $P' \subseteq 332kA$, and the query time is $O(k \log n)$.
\end{proof}

\subsection{Reducing the Query Time by Compact Encoding}

We briefly explain how to reduce the query time of the witness in Lemma~\ref{lem:cons-ap-by-cardinality-long}  from $O(k \log n)$ to $O(\min\{\frac{m}{n}, n\}\log n + \log m)$. 

The term $k$ in the original query time mainly results from the fact that a solution can have $k$ integers. But if some integer appears in a solution multiple times, we can encode them compactly by storing the number of times it appears. With such compact encoding, all the operations with solutions can be done within time linear in the number of distinct integers in the solutions. Note that a solution can have $\min\{k ,n\}$ integers. Therefore, the query time can be reduced to $O(\min\{k, n\}\log n + \log m)$. The term $\log m$ results from the $\log \frac{m}{n}$ iterations we use to augment $P$. 

We can further reduce the query time to $O(\min\{\frac{m}{n}, n\}\log n + \log m)$ by assuming that $k = \lceil \frac{m+1}{n} \rceil$. Note that $0 \in A$. For any integer $k'$ large than $k$, we have that 
\(
    kA \subseteq k'A.
\)
Therefore, if an arithmetic progression $P \subseteq kA$, then $P \subseteq k'A$. The witness for $P \subseteq kA$ can also be used as a witness for $P \subseteq k'A$. 

\thmka*

\section{Long Arithmetic Progressions in Subset Sums}\label{sec:ss}
We first rewrite Lemma~\ref{lem:cons-ap-by-cardinality-long} so that it can be conveniently used later.

\begin{lemma}\label{lem:ka}
    Let $A \subseteq \mathbb{Z}[0,m]$ be a set of $n$ integers. Assume that $0 \in A$ and that $n \geq 2$.  Let $k$ be a positive integer. Assume that
    \[
            k \geq \frac{996 m }{n}.
    \]  
    In $O(n\log m)$ time, we can compute an arithmetic progression 
    \[
        P = \{s\} + \{0, d, 2d, \ldots, md\} \subseteq kA
    \]
    with $d = \gcd(A)$, and a witness with $O(k \log n)$ expected query time.
\end{lemma}
\begin{proof}
    We assume that $d = 1$ since otherwise we can divide every integer in $A$ by $d$.  We further assume that $n \leq m + 1$ since otherwise the lemma holds trivially. Then we have 
    \[
        k \geq \frac{996m}{n} \geq 332\cdot \frac{3m}{n} \geq 332\cdot \left\lceil\frac{m+1}{n}\right\rceil.
    \]
    Then the lemma follows by Lemma~\ref{lem:cons-ap-by-cardinality-long}.
\end{proof}

The main difference between sumset $kA$ and subset sums $\mathcal{S}(A)$ is that each element of $A$ can be used up to $k$ times to obtain an integer in $kA$, but can be used at most once to obtain an integer in $\mathcal{S}(A)$. To tackle this issue, S{\'a}rk{\"o}zy~\cite{Sar94} (and also Lev~\cite{Lev03}) considered $A + A$. They showed that when $n \geq \Omega(\sqrt{m\log m})$, there is a set $B \subseteq A + A$ such that every $b \in  B$ can be represented as a sum of two integers of $A$ in many disjoint ways. In other words, each integer $b \in B$ can be used many times.  Then they showed for some integer $u$, the sumset $uB$ has an arithmetic progression of length $m$, and each integer in $uB$ can be represented as a sum of $2u$ distinct integers from $A$. As a result, $\mathcal{S}(A)$ contains an arithmetic progression of length $m$.

Given Lemma~\ref{lem:ka}, S{\'a}rk{\"o}zy's proof can be made constructive directly. However, the construction may take as much as $\Theta(n^2)$ time because the set $B$ may contain up to $\Theta(n^2)$ integers. 

To reduce the construction time to $O(n \log n)$, we need to find a set that has a similar effect as $B$ but is much smaller. We consider a set $G \subseteq A - A$ such that every $g \in G$ can be represented as a gap of two elements of $A$ in many disjoint ways and that $g \leq O(\frac{m}{n})$. Unlike the elements of $B$, which can be as large as $2m$, the elements of $G$ is a factor of $n$ smaller than $m$. As a result, $G$ needs only $O(n)$ integers in order to produce an arithmetic progression in $\mathcal{S}(A)$.  We formalize this idea in Subsection~\ref{sec:sumset2ss}.

The disadvantage of considering $G \subseteq A - A$ is that we can obtain an arithmetic progression $P$ of length only $\frac{m}{n}$. As a consequence, we have to augment $P$. As in the sumset case, we will iteratively augment $P$ using Lemma~\ref{lem:aug-idea-1} and Lemma~\ref{lem:aug-method-2} until its length becomes $m$. The difficulty is that each element of $A$ can be used for at most once during the augmentation. We shall tackle this issue in Subsection~\ref{sec:aug-ss}

\subsection{Generating Short Arithmetic Progressions}\label{sec:sumset2ss}
In this subsection, we will use Lemma~\ref{lem:ka} to prove that $\mathcal{S}(A)$ contains an arithmetic progression of length roughly $\frac{m}{n}$, which can be summarized by the following lemma.
\begin{restatable}{lemma}{lemshortapss}\label{lem:short-ap-ss}
    Let $A \subseteq \mathbb{Z}[1,m]$ be a set of $4n$ integers. Let $\ell$ be an integer with
    \[
        \frac{m}{n}  \leq \ell \leq \frac{n}{1000\log 2n}
    \]
    In $O(n \log n)$ time, we can compute a subset $A^* \subseteq A$ with $|A^*| \leq 2000\ell$, an arithmetic progression
    \[
        \{s\} + \{0, d, 2d, \ldots, \ell d\} \subseteq \mathcal{S}(A^*)
    \]
    with $d \leq \frac{m}{n}$, and a witness with $O(\ell)$ expected query time.
\end{restatable}

\subsubsection{From Sumsets to Subset Sums via Integer Pairs}

To prove this lemma, we shall consider integer pairs and their gaps. As one will see in this section, we can connect sumsets and subset sums by considering integer pairs.
\begin{restatable}{definition}{defpair}
    Let $T \subseteq \mathbb{Z} \times \mathbb{Z}$ be a set of integer pairs.  We say that $T$ is conflict-free if no two pairs in $T$ share a common integer. We define two sets as follows.
    \begin{align*}
        A_T &= \{\textrm{$z$ : $(z, z') \in T$ or $(z', z) \in T$}\}\\
        G_T &= \{z' - z : (z, z') \in T\}
    \end{align*}
    For each $g \in G_T$, the multiplicity of $g$ is define to be $|\{(z, z') \in T : z' - z = g\}|$.  We say that $G_T$ is $u$-uniform if the multiplicity of every $g \in G_T$ is $u$.
\end{restatable}
Basically, $A_T$ is the set of integers that appears in $T$. Every integer in $G_T$ can be obtained as a gap of some pair in $T$, and the multiplicity of $g \in G_T$ is the number of distinct ways that $g$ can be represented as a gap of some pair in $T$. 

Later we will show that there is a large set $T \subseteq A\times A$ that is conflict-free and $u$-uniform for some $u$. At this moment, let's consider a  set $T \subseteq \mathbb{Z} \times \mathbb{Z}$.

We make the following observation. 
\begin{observation}\label{obs:pairs}
    Let $T$ be a set of integer pairs. If $T$ is conflict-free, than $|A_T| = 2 |T|$. If $G_T$ is $u$-uniform, then $|T|=u|G_T|$.
\end{observation}

Next we show that when $T$ is conflict-free and $G_T$ is $u$-uniform, then we can shift every integer in the $u$-fold sumset $u(G_T \cup \{0\})$ by a constant value so that the resulting integers belong to $\mathcal{S}(A_T)$.  This lemma implies that if we have an arithmetic progression in the sumset $u(G_T \cup \{0\})$, then we will have an arithmetic progression of the same length in the subset sums $\mathcal{S}(A_T)$.
\begin{lemma}\label{lem:pair-to-subset-sum}
    Let $T \subseteq \mathbb{Z} \times \mathbb{Z}$ be a set of $n$ integer pairs. Suppose that $T$ is conflict-free and that $G_T$ is $u$-uniform.  Then in $O(n)$ time we can compute an integer $s_T$ such that for any $z \in u(G_T \cup \{0\})$, we have $s_T + z \in \mathcal{S}(A_T)$. Moreover, given any solution for $z \in u(G_T \cup \{0\})$, in $O(n)$ time, we can convert it to a solution for $s_T + z \in \mathcal{S}(A_T)$.
\end{lemma}
\begin{proof}
   Let $\{(a_1, b_1), \ldots, (a_n, b_n)\}$ be the pairs in $T$. Let $s_T = \sum_{i=1}^n a_i$. Consider an arbitrary integer $z \in u(G_T \cup \{0\})$. We shall show that $s_T + z \in \mathcal{S}(A_T)$. Since $z \in u(G_T \cup \{0\})$, we have that 
   \[
        z = g_1 + g_2 + \cdots + g_k
   \]
   where $g_j \in G_T$ and $k \leq u$. Since $G_T$ is $u$-uniform, every $g \in G_T$ can be represented in $u$ distinct ways as a gap of some pair in $T$.  Therefore, we can find $k$ distinct pairs $\{(a_{i_1}, b_{i_1}), \ldots, (a_{i_k}, b_{i_k})\}$ from $P$ such that $g_j = b_{i_j} - a_{i_j}$ for every $j \in \mathbb{Z}[1, k]$. That is,
   \[
        z = (b_{i_1} - a_{i_1}) + \cdots + (b_{i_k} - a_{i_k}).
   \]
   Let $I = \{i_1, \ldots, i_k\}$. Now consider $s_T + z$. 
   \[
        s_T + z = \sum_{i = 1}^n a_i + (b_{i_1} - a_{i_1}) + \cdots + (b_{i_k} - a_{i_k}) = \sum_{i\notin I} a_i + \sum_{i\in I} b_i.
   \]
   Since $T$ is conflict-free, we have that all the $2n$ integers $\{a_1, b_1, a_2, b_2, \ldots, a_n, b_n\}$ are distinct.  Therefore, 
   \[
        s_T + z \in \mathcal{S}(A_T).
   \]

   In the above argument, we actually convert a solution $(g_1, \ldots, g_k, 0, \ldots, 0)$ for $z \in u(G_T \cup \{0\})$ to a solution $\sum_{i\notin I} a_i + \sum_{i\in I} b_i$ for $s + z \in \mathcal{S}(A_T)$. It is not hard to see that the conversion can be done in $O(n)$ time.
\end{proof}

\subsubsection{Arithmetic Progressions Via Integer Pairs}

We shall prove that when $G_T$ has large cardinality and is $u$-uniform for  some $u$, then $u(G_T\cup\{0\})$ contains arithmetic progressions. As a result, $\mathcal{S}(A_T)$ also contains arithmetic progressions.

We first show that at the cost of a logarithmic factor, we can always make $G_T$ uniform for some $u$.

\begin{lemma}\label{lem:uniform}
    Let $T$ be a set of $n$ integer pairs. In $O(n)$ time, we can obtain an integer $u$ and a subset $T' \subseteq T$ with 
    \[
        |T'| \geq \frac{n}{\log 2n}
    \]
    such that $G_{T'}$ is $u$-uniform.
\end{lemma}
\begin{proof}
    For any integer $u \geq 1$, define 
    \[
        G^u_{T} = \{\text{$g \in G_{T}$ :  the multiplicity of $g$ is at least $u$}\}.
    \]
    Note that $|G^u_T| = 0$ for any $u > n$. 

    We shall show that for some $u \geq 1$, 
    \(
        |G^u_{T}| \geq \frac{n}{u\log 2n}.
    \) 
    Suppose, for the sake of contradiction, that 
    \(
        |G^u_{T}| < \frac{n}{u\log 2n}
    \)
    for all $u \geq 1$.
    This implies that 
    \[
        \sum_{u=1}^n |G^u_{T}| < \sum_{u=1}^n \frac{n}{u\log 2n} = \frac{n}{\log 2n}\sum_{u=1}^n \frac{1}{u} \leq \frac{n}{\log 2n} \cdot \log 2n = n
    \]
    But by the definition of $G^u_{T}$, we have
    \[
        \sum_{u=1}^n |G^u_{T}| = \sum_{u=1}^{n} u(|G^u_{T}| - |G^{u+1}_{T}|) = |T| = n.
    \]
    Contradiction.

    Let $u$ and $G^u_{T}$ be such that 
    \[
        |G^u_{T}| \geq \frac{n}{u\log 2n}.
    \]
    They can be computed in $O(n)$ time.  Then we can form a $u$-uniform set $T' \subseteq T$ as follows. For each $g \in G^u_{T}$, we pick $u$ distinct pairs in $T$ whose gaps are $g$. Then we have that 
    \[
        |T'| \geq u|G^u_T| \geq u\cdot \frac{n}{u\log 2n} = \frac{n}{\log 2n}. \qedhere
    \]
\end{proof}

Using Lemma~\ref{lem:pair-to-subset-sum} and Lemma~\ref{lem:uniform}, we can show that when $T$ is conflict-free and its cardinality is large when compared with the maximum gap of the pair, $\mathcal{S}(A_T)$ contains a long arithmetic progression.

\begin{lemma}\label{lem:ap-by-pairs}
    Let $T \subseteq \mathbb{Z} \times \mathbb{Z}$ be a conflict-free set of $n$ integer pairs. Suppose that $G_T \subseteq [1, g]$ for some integer $g$ and that 
    \begin{equation}\label{eq:ap-by-pairs-1}
        g \leq \frac{n}{1000\log 2n}.
    \end{equation}
    In $O(n \log n)$ time, we can compute a subset $T^* \subseteq T$ of at most $1000g$ pairs, an arithmetic progression 
    \[
        \{s\} + \{0, d, 2d, \ldots, g d\} \subseteq \mathcal{S}(A_{T^*})
    \]
    with $d  \leq \max(G_T)$, and a witness with $O(g)$ expected query time.
\end{lemma}
\begin{proof}
    By Lemma~\ref{lem:uniform}, in $O(n)$ time, we can obtain an integer $u$ and a subset $T' \subseteq T$ such that $G_{T'}$ is $u$-uniform and that
    \begin{equation}\label{eq:ap-by-pairs-2}
        |T'| \geq \frac{n}{\log 2n}.
    \end{equation}
    Note that $T'$ is conflict-free as it is a subset of $T$. By Observation~\ref{obs:pairs},
    \begin{equation}\label{eq:ap-by-pairs-3}
        |T'| = u|G_{T'}|. 
    \end{equation}
    Let $k = \lfloor \frac{1000g}{|G_{T'}|} \rfloor$. In view of~\eqref{eq:ap-by-pairs-1} and~\eqref{eq:ap-by-pairs-2},
    \[
        k \leq \frac{1000g}{|G_{T'}|} = \frac{1000g}{|T'|} \cdot u \leq u.
    \]
    We can obtain a $k$-uniform subset $T^* \subseteq T'$ as follows. For each $g \in G_{T'}$, by selecting $k$ distinct pairs from $T'$ whose gaps are $g$.  It is easy to see that $G_{T^*} = G_{T'}$. Moreover, $T'$ must be conflict-free. By Observation~\ref{obs:pairs},
    \[
        |T^*| = k|G_{T'^*}| \leq \frac{1000g}{|G_{T^*}|} \cdot |G_{T^*}| \leq 1000g.
    \]

     Also, note that 
    \[
        k \geq \frac{1000g}{|G_{T^*}|} - 1 \geq \frac{996g}{|G_{T^*}|}.
    \]
    Let $G = G_{T^*} \cup\{0\}$. By Lemma~\ref{lem:ka}, we can compute an arithmetic progression 
    \[
        \{s'\} + \{0, d, \ldots, g d\} \subseteq kG
    \]
    and a witness $\mathcal{W}$ for this arithmetic progression. Note that $d = \gcd(G) \leq \max(G) \leq \max(G_T)$. The construction time is 
    \[
        O(|G|\log g) \leq  O(n \log n),
    \]
    and the expected query time of $\mathcal{W}$ is
    \[
        O(k \log |G|) = O(\frac{g}{|G|} \log |G|) \leq O(g).
    \]
    Then by Lemma~\ref{lem:pair-to-subset-sum}, in $O(|T^*|)$ time, we can compute an integer $s_{T^*}$ such that 
    \[
        \{s' + s_{T^*}\} + \{0, d, \ldots, g d\} \subseteq \mathcal{S}(A_{T^*}).
    \]
    Moreover, we can use $\mathcal{W}$ to obtain a solution for $s' + s_{T^*} + jd$ as follows: first use $\mathcal{W}$ to obtain a solution for $s' + jd \in kG$, and then convert it to a solution for $s' + s_{T*} + jd \in \mathcal{S}(A_{T^*})$ via Lemma~\ref{lem:pair-to-subset-sum}. The first step takes $O(g)$ time, and the second step takes $O(T^*) = O(g)$ time. So the total query time is $O(g)$.
\end{proof}

\subsubsection{Short Arithmetic Progressions in Subset Sums}

Now we shall prove the main lemma of this subsection. We first show that there is a large conflict-free set $T \subseteq A\times A$ and the gaps of the pairs in $T$ are small. Then we use this set $T$ to produce arithmetic progressions.

\begin{lemma}\label{lem:gen-pairs}
    Let $A \subseteq \mathbb{Z}[1,m]$ be a sorted set of $4n$ integers. In $O(n)$ time, we can compute a conflict-free set $T \subseteq A \times A$ of at least $n$ pairs with $G_T \subseteq \mathbb{Z}[1, \frac{m}{n}]$.
\end{lemma}
\begin{proof}
    Let $a_1, \ldots, a_{4n}$ be the elements of $A$ labeled in increasing order. For $i \in \mathbb{Z}[1, 4n-1]$, define $g_i = a_{i+1} - a_i$. Let $I = \{i : g_i \leq \frac{m}{n}\}$ and let $\overline{I} = \{i : g_i > \frac{m}{n}\}$. Note that 
    \[
        \sum_{i=1}^{n-1} g_i = a_n - a_1 \leq m.
    \]
    Therefore, $\overline{I} < n$, which implies that $|I| > 3n - 1 \geq 2n$. Consider the following two sets of pairs.
    \begin{align*}
        T_1 &= \{\text{$(a_i, a_{i+1})$ : $i \in I$ and $i$ is odd} \}.\\
        T_2 & = \{\text{$(a_i, a_{i+1})$ : $i \in I$ and $i$ is even}\}
    \end{align*}
    Both $T_1$ and $T_2$ are conflict-free, and both $G_{T_1}$ and $G_{T_2}$ are subsets of $\mathbb{Z}[1, \frac{m}{n}]$. Moreover, since $|I| \geq 2n$, at least one of $T_1$ and $T_2$ have cardinality at least $n$.  We pick the one with a larger cardinality and denote it as $T$.
\end{proof}

\lemshortapss*
\begin{proof}
    We first sort $A$, which takes $O(n\log n)$ time.  By Lemma~\ref{lem:gen-pairs}, in $O(n)$ time, we can compute a conflict-free pair $T \subseteq A \times A$ such that $|T| = n$ and that $G_T \subseteq \mathbb{Z}[1, \frac{m}{n}]$. Since $\ell \geq m/n$, the set $G_T$ can be viewed as a subset of $\mathbb{Z}[1, \ell]$. Also, note that 
    \[
        \frac{n}{1000\log 2n} \geq \ell.
    \]
    By Lemma~\ref{lem:ap-by-pairs}, in $O(n \log n)$ time, we can compute a subset $T^* \subseteq T$ of at most $1000\ell$ pairs, an arithmetic progression 
    \[
        \{s\} + \{0, d, 2d, \ldots, \ell d\} \subseteq \mathcal{S}(A_{T^*})
    \]
    with $d  \leq \max({G_T}) \leq \frac{m}{n}$, and a witness with $O(\ell)$ expected query time. Note that $A_{T^*} \subseteq A$ and that $|A_{T^*}| = 2|T^*| \leq 2000\ell$. This completes the proof.
\end{proof}

\subsection{Augmenting Arithmetic Progressions Using Different Integers}\label{sec:aug-ss}
We compare Lemma~\ref{lem:short-ap-ss} with Theorem~\ref{thm:ss}. The length of the arithmetic progression produced by Lemma~\ref{lem:short-ap-ss} is a factor of $n$ smaller than that in Theorem~\ref{thm:ss}. So we shall augment it. 

Let $P = \{s\} + \{0, d, 2d, \ldots, \ell d\}$ be an arithmetic progression that to be augmented. On a high level, we shall iteratively augment $P$ using Lemma~\ref{lem:aug-idea-1} and Lemma~\ref{lem:aug-method-2} until its length becomes $m$. As in the sumset case, in order to apply Lemma~\ref{lem:aug-idea-1}, we should find pairs of integers whose gaps are not divisible by $d$, and to apply Lemma~\ref{lem:aug-method-2}, we should find pairs whose gaps are divisible by $d$. Since now we are considering the subset sums case, all the pairs should be conflict-free.

We first show that there is either a pair whose gap is small and not divisible by $d$ or a pair whose gap is large and divisible by $d$.
\begin{lemma}\label{lem:one-aug-pair}
    Let $A \subseteq \mathbb{Z}[1, m]$ be a sorted set of $n$ integers. Assume that $n \geq 4$. Let $d$ be a positive integer. Let $\ell$ be an integer with $\ell \geq \frac{16000m}{n}$. In $O(n)$ time, we can find a pair $(a, a') \in A \times A$ satisfying one of the two following conditions.
    \begin{enumerate}[label={\normalfont (\roman*)}]
        \item $d \nmid (a'  - a)$ and $1 \leq a' - a \leq \frac{\ell}{4000}$

        \item $d \mid (a' - a)$ and $ \frac{\ell d}{8000\gamma} \leq a' - a \leq \frac{\ell d}{4000}$, where $\gamma = \frac{m}{n} + \frac{\ell}{4000n}$.
    \end{enumerate} 
\end{lemma}
\begin{proof}
    Let $\ell' = \frac{\ell}{4000}$. Note that $\ell' \geq \frac{4m}{n}$. We shall prove the lemma with $\frac{\ell}{4000}$ replaced with $\ell'$. Label the elements of $A$ as $a_1, \ldots, a_{n}$ in increasing order. For $i \in \mathbb{Z}[1, n-1]$, define $g_i = a_{i+1} - a_i$. Let $\alpha$ be the number of $g_i$'s that are at most $\ell'$, and $\beta$ be the number of $g_i$'s that are at least $\ell' + 1$. Then we have that
    \[
        \beta \leq \frac{m}{\ell' + 1}.
    \]
    These $\beta$ large gaps separate the $\alpha$ small gaps into at most $\beta + 1$ groups, and at least one group has cardinality at least
    \[
        \frac{\alpha}{\beta + 1} = \frac{n - 1- \beta}{\beta + 1} = \frac{n}{\beta+1} - 1 \geq \frac{n(\ell' + 1)}{m + \ell' + 1} - 1 \geq \frac{n\ell'}{m + \ell'} - 1 \geq \frac{n\ell'}{2(m + \ell')} \geq \frac{\ell'}{2\gamma}.
    \]
    The second inequality from the end is due to that $\ell' \geq 4m/n$ and $n \geq 4$.  Let $\{g_{i+1}, \ldots,  g_{i + k}\}$ be this group. Note that all the gaps in this group is at most $\ell'$. If any of them is not divisible by $d$, then we are done as we have a pair satisfying condition (i). In the following, we assume that they are all divisible by $d$.  Then $g_{i+j} \geq d$ for any $j \in \mathbb{Z}[1, k]$ since they are all positive.  If any $g_{i + j} \geq \frac{\ell' d}{4\gamma}$, then we are also done as the pair $(a_{i+j}, a_{i+j+1})$ satisfies condition (ii). Suppose that they are less than $\frac{\ell' d}{4\gamma}$. Let $j^*$ be the minimum integers such that 
    \[
        g_{i+1} + g_{i+2} + \cdots + g_{i + j^*} \geq \frac{\ell' d}{2\gamma}.
    \]
    Such $j^*$ must exist since $g_{i+j} \geq d$ and $k \geq \frac{\ell'}{2\gamma}$. Due to the minimality of $j^*$, the above sum is at most $\frac{\ell' d}{\gamma} \leq \ell' d$.  This implies that the pair $(a_{i+1}, a_{i + j^* + 1})$ satisfies condition (ii).
\end{proof}

\smallskip
\noindent\textbf{Remark.} We observe that the running time $O(n)$ in Lemma~\ref{lem:one-aug-pair} is actually due to computing all the gaps $a_{i+1}-a_i$. Suppose all the gaps are pre-computed, together with a data structure $\mathcal{D}$ that stores the following information: whether each gap is no more than $\ell'$; for each gap no more than $\ell'$, whether it is divisible by $d$; for each gap no more than $\ell'$ and divisible by $d$, whether it is less than $\frac{\ell' d}{4\gamma}$; and moreover, for consecutive indices of gaps satisfying the last condition, it also stores the smallest and largest indices.  Then it is easy to verify that with $\mathcal{D}$ finding the desired pair $(a,a')\in A\times A$ only takes $O(1)$ time.

When $A$ has lots of integers, then we can iteratively extra pairs of integers from $A$ via Lemma~\ref{lem:one-aug-pair}.

\begin{lemma}\label{lem:aug-pairs}
    Let $d$, $\ell$, and $n$ be three positive integers.  Let $A \subseteq \mathbb{Z}[1, m]$ be a sorted set of at least $n + 40000d\log d + 8000\gamma$ integers, where $\gamma = \frac{m}{n}+ \frac{\ell}{4000n}$.  Assume that $\ell \geq \frac{16000m}{n}$.
    In $O(|A|)$ time, we can obtain a conflict-free set $T \subseteq A \times A$ satisfying one of the following.
    \begin{enumerate}[label={\normalfont (\roman*)}]
        \item $|T| = 20000 d\log d$, $G_T \subseteq \mathbb{Z}[1, \frac{\ell}{4000}]$, and $d \nmid g$ for any $g \in G_T$.

        \item $|T| = 4000\gamma$, $G_T \subseteq \mathbb{Z}[\frac{\ell d}{8000\gamma}, \frac{\ell d}{4000}]$, and $d \mid g$ for any $g \in G_T$
    \end{enumerate} 
\end{lemma}
\begin{proof}
    We can iteratively extract pairs from $A$ via Lemma~\ref{lem:one-aug-pair} until the size of $A$ becomes $n$.  Every time we obtain a pair $(a, a')$, we remove $a$ and $a'$ from $A$. Therefore, the set of pairs we obtain must be conflict-free. The pair returned in each iteration either contributes to (i) or contributes to (ii), hence we end up with a sufficiently large $T$ satisfying either (i) or (ii). It remains to analyze the overall running time.
    
    If in each iteration, we compute a pair from scratch, then we could spend as much as $O(|A|^2)$ time in total. To reduce the running time, we can use the data structure $\mathcal{D}$ according to the remark above.
    The crucial observation is that whenever a pair of integers is extracted from $A$, the gaps of the remaining integers, together with its data structure can be computed in $O(1)$ time. Thus, only the first pair costs $O(|A|)$ time, and each of the other pairs costs only $O(1)$ time. 
\end{proof}

If Lemma~\ref{lem:aug-pairs} ends with case (ii), then we can use the pairs in $T$ to augment $P$ via Lemma~\ref{lem:aug-method-2}. But if it ends with case (i), we should first use the pairs in $T$ to generate the set $Q$ required by Lemma~\ref{lem:aug-idea-1}, and then use $Q$ to augment $P$.
To generate the set $Q$, we need the following lemma.

\begin{restatable}{lemma}{lemremainder}\label{lem:remainder}
    Let $d$ be a positive integer. Let $T \subseteq \mathbb{Z} \times \mathbb{Z}$ be a conflict-free set of $4n$ integer pairs. Assume that $G_T \subseteq [1, g_{\max}]$ for some integer $g_{\max}$ and that 
    \[
        d \leq \frac{n}{1000\log 2n}.
    \]
    Also, assume that $d \nmid g$ for any $g \in G_T$. In $O(n)$ time, we can compute a proper divisor $d'$ of $d$, two integers $s$ and $h$, and a subset $T^* \subseteq T$ of at most $1000d$ pairs such that $\mathcal{S}(A_{T^*})$ contains $\frac{d}{d'}$ integers $q_0, \ldots, q_{\frac{d-d'}{d'}}$ satisfying 
    \[
        h \leq q_i \leq h + 1000dg_{\max} \quad \text{and} \quad q_i \equiv s + id' \pmod d
    \]
    for every $i = 0, 1, \ldots, \frac{d -d'}{d'}$. Within the same time, we can build a data structure that, given any $i$, returns $q_i$ and a solution $q_i \in \mathcal{S}(A_{T^*})$ in $O(d)$ time.
\end{restatable}
Comparing Lemma~\ref{lem:remainder} with Lemma~\ref{lem:ap-by-pairs}, one can find that these lemmas share a common flavor. Indeed, Lemma~\ref{lem:remainder} considers the remainder of integers modulo $d$, and can be regarded as a modular version of Lemma~\ref{lem:ap-by-pairs}.  We leave the proof of Lemma~\ref{lem:remainder} to Appendix~\ref{app:remainder}.

With Lemma~\ref{lem:aug-pairs} and Lemma~\ref{lem:remainder}, we can augment $P$ via Lemma~\ref{lem:aug-idea-1} and Lemma~\ref{lem:aug-method-2} when $|A|$ is large enough.

\begin{lemma}\label{lem:aug-one-step}
    Let $P = \{s\} + \{0, d, 2d, \ldots, \ell d\}$ be an arithmetic progression. Let $A$ be a sorted set of at least $n + 40000d\log d + 8000\gamma$ integers, where $\gamma = \frac{m}{n}+ \frac{\ell}{4000n}$. Assume that $\ell \geq \frac{16000m}{n}$.

    In $O(|A|)$ time, we can compute a proper divisor $d'$ of $d$, a subset $A' \subseteq A$ with $|A'| \leq 2000d + 8000\gamma$,  an arithmetic progression
    \[
        \{s'\} + \{0, d', \ldots, \ell'd'\} \subseteq P + \mathcal{S}(A').
    \]
    with $\ell' \geq \frac{3}{2}\ell$. Within the same time, we can build a witness with $O(d + \gamma)$ expected query time.
\end{lemma}
\begin{proof}
    By Lemma~\ref{lem:aug-pairs}, in $O(|A|)$ times, we can obtain a conflict-free set $T \subseteq A \times A$ satisfying one of the following.
    \begin{enumerate}[label={\normalfont (\roman*)}]
        \item $|T| = 20000 d\log d$, $G_T \subseteq \mathbb{Z}[1, \frac{\ell}{4000}]$, and $d \nmid g$ for any $g \in G_T$.

        \item $|T| = 4000\gamma$, $G_T \subseteq \mathbb{Z}[\frac{\ell d}{8000\gamma}, \frac{\ell d}{4000}]$, and $d \mid g$ for any $g \in G_T$
    \end{enumerate} 
    In case (ii), every pair can be used to augment $P$ via Lemma~\ref{lem:aug-method-2}. Since each pair can be used only once, we set the parameter $h$ in Lemma~\ref{lem:aug-method-2} to be $1$. After applying Lemma~\ref{lem:aug-method-2} for each pair once, which takes a total time of $O(\gamma) \leq O(|A|)$, we can obtain an arithmetic progression 
    \[
        \{s'\} + \{0, d, \ldots, \ell'd\} \subseteq P + \mathcal{S}(A_{T}).
    \]
    as well as a witness of $O(\gamma)$ query time. Note that $|A_T| = 2|T| = 8000\gamma$. Since each pair increases the length of $P$ by at least $\frac{\ell}{8000\gamma}$, we have that 
    \[
        \ell' \geq \ell + \frac{\ell}{8000\gamma} \cdot 4000\gamma = \frac{3}{2}\ell.
    \]

    Now consider case (i). One can verify that 
    \[
            d \leq \frac{|T|}{1000\log (2|T|)}
    \]

    By Lemma~\ref{lem:remainder}, in $O(|T| ) \leq O(|A|)$ time, we can compute a proper divisor $d'$ of $d$, two integers $s$ and $h$, and a subset $T^* \subseteq T$ of at most $1000d$ pairs such that $\mathcal{S}(A_{T^*})$ contains $\frac{d}{d'}$ integers $q_0, \ldots, q_{\frac{d-d'}{d'}}$ satisfying 
    \[
        h \leq q_i \leq h + \frac{\ell d}{4} \quad \text{and} \quad q_i \equiv s + id' { \pmod d}
    \]
    for every $i = 0, 1, \ldots, \frac{d -d'}{d'}$. In the same time, we can build a data structure that, given any $i$, returns $q_i$ and a solution $q_i \in \mathcal{S}(A_{T^*})$ in {$O(d)$} time.

    Then we use $\mathcal{S}(A_{T^*})$ to augment $P$. Applying Lemma~\ref{lem:aug-idea-1} (by let $Q = \mathcal{S}(A_{T^*})$, $h_{\min} = \lfloor\frac{h - s}{d}\rfloor$, and $h_{\max} = \lfloor\frac{h + \frac{\ell d}{4} - s}{d}\rfloor$), we compute an arithmetic progression
    \[
        \{s'\} + \{0, d', \ldots, \ell'd'\} \subseteq P + \mathcal{S}(A_{T^*}).
    \]
    as well as a witness of $O(d)$ query time.  We have that
    \[
        \ell' \geq (\ell - h_{\max} + h_{\min}) \cdot \frac{d}{d'} \geq (\ell - \frac{\ell}{4} ) \cdot \frac{d}{d'} \geq \frac{3}{2}\ell.
    \]
    The last inequality is due to that $d'$ is a proper divisor of $d$.
    Note that $|A_{T^*}| = 2|T^*| \leq 2000d$.
\end{proof}

\subsection{Putting Things Together}
Now we are to prove Theorem~\ref{thm:ss} using the results of the previous two subsections.
\begin{lemma}\label{lem:long-ap-ss}
    Let $A \subseteq \mathbb{Z}[1,m]$ be a set of $6n$ integers. For any integer $\ell$ with 
    \[
        m \leq \ell \leq \frac{n^2}{10^7\log 2n}
    \]
    In $O(n \log n)$ time, we can compute a subset $A^* \subseteq A$ with $|A^*| \leq \frac{30000\ell\log n}{n}$, an arithmetic progression
    \[
        \{s\} + \{0, d, 2d, \ldots, \ell d\} \in \mathcal{S}(A^*)
    \]
    with $d \leq \frac{m}{n}$, and a witness with $O(\frac{\ell \log n}{n})$ expected query time for this arithmetic progression.
\end{lemma}
\begin{proof}
    We first partition $A$ into $A'$ and $A''$ with $|A'| = 4n$ and $|A''| = 2n$. Let $\ell_0 = \frac{16000\ell}{n}$. We will use $A'$ to produce an arithmetic progression of length $\ell_0$, and use $A''$ to iteratively augment it via Lemma~\ref{lem:aug-one-step} until its length becomes at least $\ell$. 

    One can verify that 
    \[
        \frac{16000m}{n} \leq \ell_0 \leq \frac{n}{1000\log 2n}.
    \]
    By Lemma~\ref{lem:short-ap-ss},  in $O(n \log n)$ time, we can obtain a subset $A_0 \subseteq A'$ with $A_0 \leq 2000\ell_0$, an arithmetic progression 
    \[
        P_0 = \{s_0\} + \{0, d_0, \ldots, \ell_0d_0\} \subseteq \mathcal{S}(A_0)
    \]
    with $d_0 \leq \frac{m}{n}$, and a witness $\mathcal{W}_0$ with $O(\ell_0)$ expected query time.

    Then we shall augment the arithmetic progression using $A''$. At this stage, let's assume that Lemma~\ref{lem:aug-one-step} is always applicable, and we will prove this later. Given an arithmetic progression $P_i = \{s_i\} + \{0, d_i, \ldots, \ell_i d_i\}$, if $\ell_i < \ell$, we can compute a subset $A_{i+1}$ of $A''$ and use it to augment $P_i$  to obtain an arithmetic progression $P_{i+1}$ with $\ell_{i+1} \geq \frac{3}{2}\ell_i$ via Lemma~\ref{lem:aug-one-step}. We can also obtain a witness $\mathcal{W}_{i+1}$ for $P_{i+1} \subseteq P_i + \mathcal{S}(A_{i+1})$.  To avoid conflict, $A_i$ will be removed from $A''$ after the augmentation. Let $P_1, \ldots, P_h$ be the arithmetic progression we obtained. Since $\ell_0 = \ell/n$, after at most $2\log n$ iterations, the length of the arithmetic progression will be greater than $\ell$. Therefore, we have that $h \leq 2\log n$. 

    Let $A^* = A_0 \cup A_1 \cdots A_h$. Then we have that 
    \[
        P_h \subseteq \mathcal{S}(A_0) + \mathcal{S}(A_1) + \cdots + \mathcal{S}(A_h) = \mathcal{S}(A^*).
    \]
    Note that for $i \in \mathbb{Z}[0,h]$,
    \[
        d_i \leq d_0 \leq \frac{m}{n} \leq \frac{\ell}{n}
    \]
     and that for $i = \mathbb{Z}[0, h-1]$,
     \[
        |A_{i+1}| \leq 2000d_i + \frac{8000m}{n} + \frac{8000\ell_i}{4000n} \leq  \frac{10002\ell}{n}.
     \]
    So we have that 
    \[
        |A^*| = |A_0| + \cdots + |A_h| \leq 2000\ell + \frac{10002\ell}{n} \cdot 2\log n \leq \frac{30000\ell}{n}\cdot \log n.
    \]
    As to the witness, $\mathcal{W}_0, \ldots, \mathcal{W}_h$ together form a witness for $P_h \subseteq \mathcal{S}(A')$. The total query time is 
    \[
        \ell_0 + \sum_{i=1}^h (d_i + \frac{m}{n} + \frac{\ell_i}{n}) \leq O(\frac{\ell}{n}\cdot h) = O(\frac{\ell}{n} \log n).
    \]

    It remains to show that $A''$ indeed has enough integers so that Lemma~\ref{lem:aug-one-step} is always applicable when obtaining $P_1, \ldots,  P_h$. Recall that 
    \[
        |A^*| \leq \frac{30000\ell}{n}\cdot \log n \leq \frac{n}{300}.
    \]
    Therefore, $A''$ has at least $2n - \frac{n}{300} \geq \frac{3}{2}n$ integers during the procedure. Also, note that 
    \[
        n + 40000d_i\log d_i + \frac{8000m}{n} + \frac{8000\ell_i}{4000n} \leq n + 50000\frac{\ell}{n} \cdot \log n \leq n + \frac{n}{200}\leq \frac{3n}{2}   
    \]
    and that 
    \[
        \ell_i \geq \ell_0 \geq \frac{16000\ell}{n} \geq \frac{16000m}{n}.
    \]
    These imply that Lemma~\ref{lem:aug-one-step} is always applicable when we obtain $P_1, \ldots, P_h$.
\end{proof}

\thmss*
\begin{proof}
    Since
    \[
        m \leq \ell \leq \frac{n^2}{5 \times 10^8\log(2n)},
    \]
    we have that $n$ is sufficiently large so that $\lfloor\frac{n}{6} \rfloor \geq \frac{n}{7}$.  The theorem followings by replacing the $n$ in Lemma~\ref{lem:long-ap-ss} with $\lfloor\frac{n}{6} \rfloor$.
\end{proof}

\section{Applications}\label{sec:application}
\subsection{Unbounded Subset Sum}
We shall prove the following theorem.
\thmunbounded*

\begin{proof}
    Let $A_{n-1}=\{0,a_1,\cdots,a_{n-1}\}$ and $d=\gcd(A_{n-1})$, then it is straightforward that $\gcd(d,a_{n})=1$. Further, $d$ can be computed in $O(n\log a_{1})$ time. Notice that for any $1\le i\le n-1$, $a_i-a_{i-1}>0$ is always a multiple of $d$. Thus $a_{n-1}\ge (n-1)d$, which means
    \begin{eqnarray}\label{eq:unbounded-1}
        d\le \frac{a_{n-1}}{n-1}.   
    \end{eqnarray}

    Denote by $\frac{A_{n-1}}{d}=\{0,a_1/d,a_2/d,\cdots,a_{n-1}/d\}$. Apply Theorem~\ref{thm:ka} with $A=\frac{A_{n-1}}{d}$, $m=a_{n}-1$ and $k=\lceil\frac{a_n}{n}\rceil$, then within $O(n\log n+\log a_n)$ time we can compute some integer $s$ such that 
    \begin{eqnarray}\label{eq:unbounded}
    \{s\}+\{0,1,\cdots,a_{n}-1\}\subseteq 332k\frac{A_{n-1}}{d},    
    \end{eqnarray}  
    together with a witness. 
    
    Let $R=\{0,a_n,2a_n,\cdots,(d-1)a_n\}$. Since $\gcd(d,a_n)=1$, it is clear that $ia_n\not\equiv i'a_n \pmod d$ for $0\le i<i'\le d-1$, thus $R\equiv \{0,1,\cdots,d-1\}\pmod d$. Moreover, we can compute every $ia_n \pmod d$ in $O(n\log d)\le O(n\log a_1)$ time.
    
    Let $t_0 = (d-1)a_n+sd$. Consider arbitrary $t\geq t_0$. We first compute $t \bmod d$ and find some $0\le i_t\le d-1$ such that $i_ta_n\equiv t \pmod d$. Hence, $d|(t-i_ta_n)$, and moreover, $\frac{t-i_ta_n}{d}\ge s$. Consequently, for some $0\le r\le a_n-1$ and $q\in\mathbb{N}$, we have that
    $$\frac{t-i_ta_n}{d}-s=r+qa_n.$$
Thus,
$$\frac{t-i_ta_n}{d}\in \{qa_n+s\}+\{0,1,\cdots,a_n-1\}\subseteq \{qa_n\}+332k\frac{A_{n-1}}{d},$$
or equivalently,
$${t}\in \{(qd+i_t)a_n+sd\}+\{0,d,\cdots,(a_n-1)d\}\subseteq \{(qd+i_t)a_n\}+332k{A_{n-1}},$$
    giving a solution to $\sum_i a_ix_i=t$.
    
    
    It remains to upper bound $t_0$. Notice that $s\le 332ka_{n-1}/d\le 332\lceil\frac{a_n}{n-1}\rceil a_{n-1}/d$, combining Eq~\eqref{eq:unbounded-1}, we know that
    $$t_0\le 332\lceil\frac{a_n}{n-1}\rceil a_{n-1}+\frac{a_{n-1}a_n}{n-1}\le 333\lceil\frac{a_n}{n-1}\rceil a_{n-1}.$$

Finally, we analyze the running time. It is easy to see that computing $d,i_t$ takes $O(n\log a_1)$-time, computing $q,r$ takes $O(1)$-time. Recall that computing $s$ together with a witness for Eq~\eqref{eq:unbounded} takes $O(n\log n+\log a_n)$, and the witness returns a solution per query in $O(n\log n)$-time. Thus, the overall running time is bounded by $O(n\log a_n)$.   
\end{proof}

\subsection{Dense Subset Sum}
Let $\Sigma_A = \sum_{a \in A}a$. 
\begin{theorem}
     There exists a constant $c_\delta$ such that the following is true. Given a set $A \subseteq \mathbb{Z}[1,m]$ of $n$ integers and a target $t\leq \Sigma_A/2$. If
     \[
        t\geq 0.5 c_\delta \log (2n) m\Sigma_A/n^2,
    \]
    we can decide whether $t\in \mathcal{S}(A)$ in time $\widetilde{O}(n)$, and moreover,  if the answer is "yes", we can return $B\subseteq A$ such that $\Sigma_B=t$ in time $\widetilde{O}(n)$.
\end{theorem}

The proof the theorem is quite involved. Basically one should follow the original proof in~\cite{BW21}, and analyze the time of every step of back-tracing. We defer it to Appendix~\ref{apsec:dense}. We also remark that Appendix~\ref{apsec:dense} actually proves a more general theorem that works for multi-sets.

\section{Conclusion}\label{sec:conclude}
In this paper, we show that there exists a near-linear time algorithm which computes a long arithmetic progression in sumsets and subset sums.  As an application, we obtain near-linear time algorithms for the search version of dense Subset Sum, and the search version of unbounded Subset Sum within the region $t\ge t_0=\Theta(a_n^2/n)$. 

We remark that Subset Sum admits an algorithm of running time $\widetilde{O}(n+\sqrt{a_{\max}t})$ that solves the decision version~\cite{chen2024improved}, but we are not able to obtain an algorithm with a similar running time for the search version by using our results in this paper. The main reason is that the $\widetilde{O}(n+\sqrt{a_{\max}t})$-time algorithm is based on another combinatorics result by Szemer{\'e}di and Vu~\cite{SV05}, which roughly states that given a sequence of integer sets $A_1,A_2,\cdots,A_\ell$ with $A_i\subseteq \mathbb{Z}[1,m]$, if $\sum_i |A_i|=\Omega(m)$, then there exists an arithmetic progression of length $m$ in $A_1+A_2+\cdots+A_\ell$. Though of a similar flavor, Szemer{\'e}di and Vu's theorem is proved by a completely different approach than the two finite addition theorems studied in this paper. It is far from clear whether the techniques for finite addition theorems also work there. We remark that Lev~\cite{lev2010consecutive} presented a proof for a weaker version of Szemer{\'e}di and Vu's theorem. Lev's proof shares a similar flavor as the proof for finite addition theorems, but it requires additionally that none of the $A_i$'s are contained in an arithmetic progression with difference greater than 1. It is not clear how to get rid of such a condition on $A_i$'s, meanwhile, the $\widetilde{O}(n+\sqrt{a_{\max}t})$-time algorithm cannot follow from this weaker version. It is an important open problem whether we can obtain a (near-)linear time algorithm for the Szemer{\'e}di and Vu's theorem which produces the arithmetic progression in $A_1+A_2+\cdots+A_\ell$ together with a witness.

\clearpage
\appendix
\section{A Proof for Lemma~\ref{lem:dense-half-cons}}\label{app:comp-u}
We shall prove that the integer $u$ in Lemma~\ref{lem:dense-half-exist} can be computed in $O(n \log n)$ time. Recall Lemma~\ref{lem:dense-half-exist}.
\lemdensehalf*

\densehalfcons*

\begin{proof}
Suppose that there is some $u$ satisfying case (i) of Lemma~{\ref{lem:dense-half-exist}}. We show that it can be found in $O(n \log n)$ time. Case (ii) can be tackled similarly.

We sort $A$ and label its elements as $\{a_1, \ldots, a_n\}$ in increasing order. Since $|A[u+1, u+1]| \geq \frac{1}{2r} > 0$, we have that $u+1 \in A$. Therefore, $u$ must be $a_{i} - 1$ for some $a_{i} \leq \frac{m}{2} + 1$.  Finding $u$ is equivalent to finding $a_{i^*} \in A$ such that $a_{i^*} \leq \frac{m}{2} + 1$ and that for every $v \in [a_{i^*}, m]$,
    \[
        |A[a_{i^*}, v]|\geq \frac{v - a_{i^*} + 1}{2k}, \quad \textrm{or equivalently,} \quad \frac{|A[a_{i^*}, v]|}{v - a_{i^*} + 1} \geq \frac{1}{2k}.
    \]
    For simplicity, we define $a_{n+1} = m + 1$.  It is easy to observe that the minimum value of $\frac{A[a_{i^*}, v]}{v - a_{i^*} + 1}$ must be reached when $v = a_i -1$ for some $i > i^*$ and that $|A[a_{i^*}, a_i - 1]| = i - i^*$. Therefore, finding $u$ is equivalent to finding $i^*$ such that $a_{i^*} \leq \frac{m}{2} + 1$ and that for all $i \in [i^*+1, n+1]$,
    \[
        (i - i^* ) \geq \frac{a_i - a_{i^*}}{2k}, \quad \textrm{or equivalently,} \quad 2k(i - i^* ) \geq (a_i - a_{i^*}).
    \]

For $1 \leq i \leq j \leq n+1$, we say the pair $(i, j)$ of indices is good if 
\[
    2k(j - i) \geq (a_j - a_i)
\]
and bad, otherwise. Therefore, our goal is to find the smallest index $i^*$ such that $(i^*, k)$ is good for all $k$ with $i^* \leq k \leq n+1$.  We claim that, in $O(n)$ time, Algorithm~\ref{alg:dense-half} returns the target $i^*$. 

\begin{algorithm}
\caption{An Linear-Time Algorithm to Find $i^*$}
\label{alg:dense-half}
\begin{algorithmic}[1]
    \Statex \textbf{Input}: $n$ integers with $0 \leq a_1 < a_2 \ldots < a_n \leq m$
    \State $a_{n+1} \leftarrow m+1$
    \State $i \leftarrow 1$, $j \leftarrow 1$
    \While{$j \leq n+1$}
        \State $j \leftarrow j + 1$
        \While{$2k(j - i) < a_j - a_i$}
            \State $i\leftarrow i + 1$
        \EndWhile
    \EndWhile
    \State \Return $i$
\end{algorithmic}
\end{algorithm}

We make the following observations.

\begin{enumerate}[label={(\alph*)}]
    \item If $(i,j)$ is bad but $(i', j)$ is good for some $i'$ with $i < i' \leq j$, then $(i, i')$ must be bad.

    \item If $(i,j)$ is good but $(i, i')$ is bad for some $i'$ with $i < i' \leq j$, then $(i', j)$ must be good.
\end{enumerate} 

To prove the correctness of the algorithm, we first prove that the outer while loop of Algorithm~\ref{alg:dense-half} maintains the invariant that for any $k$ with $i \leq k \leq j$, the pair $(i, k)$ is good. Before the first iteration, we have that $i = j = 1$, so the invariant holds. We shall show that each iteration preserves this invariant. Suppose that before an iteration, the invariant holds for $i$ and $j$. If the pair $(i, j+1)$ is good, then the current iteration does not update $i$, and clearly, the invariant holds for $i$ and $j + 1$. If $(i, j+1)$ is bad, then the current iteration updates $i$ to $i'$ so that $(i' ,j + 1)$ is good. We should show that $(i', k)$ is good for all $k$ with $i' \leq k \leq j+1$. If $k = j+1$, then this is obvious. Assume that $k \leq j$. By observation (a), the pair $(i, i')$ must be bad. Recall that the invariant before the current iteration guarantees that $(i, k)$ is good. Then by observation (b), $(i', k)$ must be good. This finishes the proof of the invariant.

Let $i$ be the index returned by the algorithm. The invariant implies that $(i, k)$ is good for all $k$ with $i^* \leq k\leq n+1$. To prove the correctness of the algorithm, it remains to show that the algorithm must return the smallest such $i$. Let $i^*$ be the smallest such index $i$. Note that the algorithm increases $i$ one by one. Therefore, it must be that $i = i^*$ in some iteration of the algorithm. But then since $(i^*, k)$ is good for all $k$ with $i^* \leq k\leq n+1$, the algorithm will never update $i$ again. Therefore, it will return $i^*$ at the end.

It is easy to see that the running time of the algorithm is $O(n)$.
\end{proof}

\section{A Proof for Lemma~\ref{lem:remainder}}\label{app:remainder}
We shall prove the following lemma in this section.
\lemremainder*

Recall the following definition. 
\defpair*

Given an integer $d$, we further define $R_T(d) = \{ (z'- z) \bmod d: (z, z' ) \in T \}$. The each $r \in R_T(d)$, the multiplicity of $r$ is defined to be the number of pairs in $T$ whose gap is congruent to $r$ modulo $d$. That is, the multiplicity of $r$ is $|\{(z, z') \in T : z' - z \equiv r \pmod d\}|$. We say that $R_T(d)$ is $u$-uniform if the multiplicity of every $r \in R_T(d)$ is $u$. The proof of Lemma~\ref{lem:remainder} is quite similar to that of Lemma~\ref{lem:ap-by-pairs}, except that now we use $R_T(d)$ instead of $G_T$.

We can always make $R_T(d)$ $u$-uniform for some $u$ at the cost of a logarithmic factor.
\begin{lemma}\label{lem:uniform-remainder}
    Let $d$ be a positive integer. Let $T$ be a set of $n$ integer pairs.  In $O(n)$ time, we can obtain an integer $u$ and a subset $T' \subseteq T$ with 
    \[
        |T'| \geq \frac{n}{\log 2n}
    \]
    such that $R_{T'}(d)$ is $u$-uniform.
\end{lemma}
We omit the proof since it is similar to that of Lemma~\ref{lem:uniform}.  

When $T$ is conflict-free and $R_T(d)$ is $u$-uniform, we can map each integer in $u(R_T(d) \cup \{0\} \cup \{d\})$ to an integer in $\mathcal{S}(A_T)$ so that certain conditions are satisfied.

\begin{lemma}\label{lem:pair-to-subset-sum-remainder}
    Let $T \subseteq \mathbb{Z} \times \mathbb{Z}$ be a conflict-free set of $n$ integer pairs. Suppose that $G_T \subseteq \mathbb{Z}[1, g]$ and that $R_T(d)$ is $u$-uniform.  Then in $O(n)$ time we can compute an integer $s_T$ such that for any $z \in u(R_T(d) \cup \{0\} \cup \{d\})$, there is some integer $q \in \mathcal{S}(A_T)$ such that $q \equiv s_T + z \pmod d$ and that $s_T \leq r \leq s_T + ug$. Moreover, given any solution for $z \in u(R_T(d) \cup \{0\} \cup \{d\})$, in $O(n)$ time, we can convert it to a solution for the corresponding $q \in \mathcal{S}(A_T)$.
\end{lemma}
\begin{proof}
    Let $\{(a_1, b_1), \ldots, (a_n, b_n)\}$ be the pairs in $T$. Let $s_T = \sum_{i=1}^n a_i$. Consider an arbitrary integer $z \in u(R_T(d) \cup \{0\} \cup \{d\})$.
    \[
        z = r_1 + r_2 + \cdots + r_u.
    \]
    Without loss of generality, assume that $r_1, \ldots, r_k$ are from $R_T(d)$, and $r_{k+1}, \ldots, r_u$ are either $0$ or $d$.  Let $z' = r_1 + \cdots + r_k$. We have that 
    \[
        z' \equiv z \pmod d.
    \]
    Since $R_{T'}(d)$ is $u$-uniform and $k \leq u$, we can find $k$ distinct pairs $(a_{i_1}, b_{i_1}), \ldots, (a_{i_k}, b_{i_k})$ from $T$ such that $r_j = (b_{i_j} - a_{i_j}) \bmod d$. Let $z'' = (b_{i_1} - a_{i_1}) + \cdots + (b_{i_k} - a_{i_k})$. We have that 
    \[
        z'' \equiv z' \pmod d \quad \text{and}\quad 0\leq z'' \leq kg \leq ug.
    \]
    Let $q = z'' + s_T$. Then $q \equiv z'' + s_T \equiv z + s_T \pmod d$ and $s_T \leq q \leq s_T + ug$. Let $I = \{i_1, \ldots, i_k\}$.  
    \[
        q = \sum_{i \notin I} a_i + \sum_{i\in I} b_i.
    \]
    Since $T$ is conflict-free, 
    \[
        q \in \mathcal{S}(A_{T^*})
    \] 

    The above argument actually converts a solution for $z \in u(R_T(d) \cup \{0\} \cup \{d\})$ to a solution for $q \in \mathcal{S}(A_{T^*})$. It is easy to see that the conversion can be done in $O(n)$ time.
\end{proof}

Now we are ready to prove Lemma~\ref{lem:remainder}.
\lemremainder*
\begin{proof}
    By Lemma~\ref{lem:uniform}, in $O(n)$ time, we can obtain an integer $u$ and a subset $T' \subseteq T$ such that $R_{T'}(d)$ is $u$-uniform and that
    \begin{equation}\label{eq:ap-by-pairs-remainder-2}
        |T'| \geq \frac{n}{\log 2n}.
    \end{equation}
    Note that $T'$ is conflict-free as it is a subset of $T$. Since $R_{T'}(d)$ is $u$-uniform,
    \begin{equation}\label{eq:ap-by-pairs-remainder-3}
        |T'| = u|R_{T'}(d)|. 
    \end{equation}
    Let $k = \lfloor \frac{1000d}{|R_{T'}(d)|} \rfloor$. In view of~\eqref{eq:ap-by-pairs-remainder-2},
    \[
        k \leq \frac{1000d}{|R_{T'}(d)|} = \frac{1000d}{|T'|} \leq \frac{1000d\log 2n}{n} \cdot u \leq u.
    \]
    The last inequality is due to that $d \leq \frac{n}{1000\log 2n}$. We can obtain a $k$-uniform subset $T^* \subseteq T'$ as follows. For each $r \in R_{T'}(d)$, by selecting $k$ distinct pairs from $T'$ whose gaps are $g$ .  It is easy to see that $R_{T^*}(d) = R_{T'}(d)$. Moreover, $T'$ must be conflict-free. Therefore,
    \[
        |T^*| = k|R_{T^*}(d)| \leq \frac{1000d}{|R_{T^*}(d)|} \cdot |R_{T^*}(d)| \leq 1000d.
    \]

    Also, note that 
    \[
        k \geq \frac{1000d}{|R_{T^*}(d)|} - 1 \geq \frac{996d}{|R_{T^*}(d)|}.
    \]
    Let $R = R_{T^*}(d) \cup\{0\} \cup \{d\}$.  One can see that $R \subseteq \mathbb{Z}[0, d]$. By Lemma~\ref{lem:ka}, we can compute an arithmetic progression 
    \[
        \{s'\} + \{0, d', \ldots, d d'\} \subseteq kR
    \]
    and a witness $\mathcal{W}$ for this arithmetic progression. Note that $d' = \gcd(R)$. Since $d \nmid g$ for any $g \in G_T$, we have that $R_{T^*}(d) \subseteq \mathbb{Z}[1, d-1]$. As a result, $R$ contains integers other than $0$ and $d$. Therefore, $d'$ must be a proper divisor of $d$. The construction time is 
    \[
        O(|R|\log d) \leq  O(n),
    \]
    and the expected query time of $\mathcal{W}$ is
    \[
        O(k \log |R|) = O(\frac{d}{|R|} \log |R|) \leq O(d).
    \]
    Then by Lemma~\ref{lem:pair-to-subset-sum-remainder}, in $O(|T^*|)$ time, we can compute an integer $s_{T^*}$ such that for any $i = 0, 1, \ldots, d^*$, there is an integer $q_i \in \mathcal{S}(A_{T^*})$,
    \[
        q_i \equiv s' + s_{T^*} + id' \quad \text{and} \quad s_{T^*} \leq r_i \leq s_{T^*} + kg \leq S_T + 1000dg.
    \]
    Moreover, given any $i \in \mathbb{Z}[0, d]$, we can use $\mathcal{W}$ to obtain $q_i$ and a solution for $q_i \in \mathcal{S}(A_{T^*})$ as follows: first use $\mathcal{W}$ to obtain a solution for $s' + id \in kR$, and then convert it to a solution for $q_i\in \mathcal{S}(A_{T^*})$ via Lemma~\ref{lem:pair-to-subset-sum-remainder}. The first step takes $O(d)$ time, and the second step takes $O(T^*) = O(d)$ time. So the total query time is $O(d)$.
\end{proof}

\section{Application: Dense Subset Sum}\label{apsec:dense}

In this section, we revisit the dense Subset Sum theorem by Bringmann and Wellnitz~\cite{BW21}. 
Roughly speaking, the dense Subset Sum theorem states that, given a Subset Sum instance with $N$ integers in $[1,m]$, if $N$ is sufficiently large compared with $m$, then within $\widetilde{O}(N)$ time we can decide whether an arbitrary target $t$ can be hit, provided that $t$ is not too far away from the middle (i.e., the sum of all integers divided by 2). See Theorem~\ref{thm:BW21} below as a formal description. Bringmann and Wellnitz explicitly stated that their algorithm only resolves the decision in near-linear time, but not the solution reconstruction (see Remark 3.7 of~\cite{BW21}). We show that, with our Theorem~\ref{thm:ss}, we can also reconstruct a solution in near-linear time. 

Towards formally describing the dense Subset Sum theorem, we need to introduce some notions together with some parameters.

Throughout this section, we are dealing with multi-sets. For a multi-set $A$ and an integer $x$, we denote by $\mu(x;A)$ the multiplicity of $x$ in $A$. A number that does
not appear in $A$ has multiplicity 0. A subset $A'\subseteq A$
is also a multi-set with $\mu(x;A')\le \mu(x;A)$ for all integer $x$.

The following notations are in line with that of~\cite{BW21}.
\begin{itemize}
    \item Size $|A|$: the number of elements in $A$, counted with multiplicity, i.e., $|A|=\sum_x \mu(x;A)$. We shall also refer to it as $N$, to distinguish it with $n$ when $A$ is a set in previous sections.
    \item Multiplicity $\mu_A$: the maximal multiplicity, i.e., $\mu_A=\max_x\mu(x;A)$.  
    \item $\Sigma_A$: the sum of all elements in $A$, i.e., $\Sigma_A=\sum_x x\mu(x;A)$. 
    \item $m_A$: the largest integer in $A$. We may also write it as $m$ if it is clear from the context which multi-set we are discussing.
    \item $\delta$-dense. A multi-set $A$ is $\delta$-dense if it satisfies that $|A|^2\ge \delta\mu_A m$.
    \item $A(d)$ and $\overline{A(d)}$. We use $A(d):=A\cap d\mathbb{Z}$ to denote the multi-set of all integers in $A$ that are divisible by $d$, and $\overline{A(d)}=A\setminus A(d)$ the multi-set of all integers in $A$ that are not divisible by $d$.
    \item Almost divisor. An integer $d>1$ is called an $\alpha$-almost divisor of $A$ if $|\overline{A(d)}|\le \alpha\mu_A\Sigma_A/|A|^2$.
\end{itemize} 
Define parameters $\alpha$, $\delta$ with respect to a multi-set $A$ as
\begin{itemize}
    \item $\alpha:=c_{\alpha}\log (2\mu_A)$;
    \item $\delta:=c_\delta \log (2N)\log^2 (2\mu_A)$;
\end{itemize}
where $c_\alpha,c_\delta$ are fixed constant.

With the above notations, we are ready to present the dense Subset Set Sum theorem in~\cite{BW21}.

\begin{theorem}[Cf. Theorem 4.6 of \cite{BW21}]\label{thm:BW21}
    Given any multi-set $A$ of $N$ integers, the following holds:
    \begin{itemize}
        \item[(i).] Let $C_\lambda:=c_\lambda\log (2\mu_A)$ for some fixed constant $c_{\lambda}$. If $A$ is $\delta$-dense, then within $\widetilde{O}(N)$ time we can preprocess $A$ such that given any query $t$ satisfying that $(4+2C_\lambda)\mu_Am\Sigma_A/N^2\le t\le \Sigma_A/2$, we can decide whether $t\in \mathcal{S}(A)$ in time $O(1)$.
        \item[(ii).] In particular, given a multi-set $A$ and a target $t\le \Sigma_A/2$ with $t\ge 0.5 \delta\mu_Am\Sigma_A/N^2$, we can decide whether $t\in \mathcal{S}(A)$ in time $\widetilde{O}(N)$. 
    \end{itemize}
\end{theorem}

Our goal in this section is to establish a similar theorem that not only resolves the decision version of dense Subset Sum, but also returns a solution in near-linear time. More precisely, 

\begin{theorem}\label{thm:dense}
    Given any multi-set $A$ of $N$ integers,  the following holds:
    \begin{itemize}
        \item[(i).] Let $C_\lambda':=c_\lambda'\log (2\mu_A)\log (2N)$ for some fixed constant $c_{\lambda}'$. If $A$ is $\delta$-dense, then within $\widetilde{O}(N)$ time we can preprocess $A$ such that given any query $t$ satisfying that $(4+2C_\lambda')\mu_Am\Sigma_A/N^2\le t\le \Sigma_A/2$, we can decide whether $t\in \mathcal{S}(A)$ in time $O(1)$, and moreover, if the answer is "yes", we can return $B\subseteq A$ such that $\Sigma_B=t$ in time $\widetilde{O}(N)$.
        \item[(ii).] In particular, given a multi-set $A$ and a target $t\le \Sigma_A/2$ with $t\ge 0.5 \delta\mu_Am\Sigma_A/N^2$, we can decide whether $t\in \mathcal{S}(A)$ in time $\widetilde{O}(N)$, and moreover,  if the answer is "yes", we can return $B\subseteq A$ such that $\Sigma_B=t$ in time $\widetilde{O}(N)$.
    \end{itemize}
\end{theorem}

\noindent\textbf{Remark on the relationship between Statement (i) and (ii) in Theorem~\ref{thm:BW21} and Theorem~\ref{thm:dense}.} 
In both theorems, the condition $\Sigma_A/2\ge t\ge 0.5 \delta\mu_Am\Sigma_A/N^2$ in $(ii)$ implies that $A$ is $\delta$-dense. Moreover, since $\delta\gg C_{\lambda}$ and $\delta\gg C_{\lambda}'$, it holds that $0.5 \delta\mu_Am\Sigma_A/N^2\ge (4+2C_\lambda)\mu_Am\Sigma_A/N^2$ and $0.5 \delta\mu_Am\Sigma_A/N^2\ge (4+2C_\lambda')\mu_Am\Sigma_A/N^2$. Hence, in both theorems, $(ii)$ follows as a direct consequence of $(i)$. In other words, it suffices to prove Theorem~\ref{thm:dense}-$(i)$.

 
 \smallskip

\noindent\textbf{Remark on $C_{\lambda}$ vs. $C_{\lambda}'$.} The $C_{\lambda}'$ in our Theorem~\ref{thm:dense} is by a logarithmic factor larger than $C_{\lambda}$, this diminishes the nontrivial region $[(4+2C_{\lambda}),\Sigma_A/2]$. In the later analysis, we shall specify where this loss comes from. Nevertheless, the nontrivial region in Statement $(ii)$, i.e., $[0.5 \delta\mu_Am\Sigma_A/N^2,\Sigma_A/2]$, is not affected. We remark that, for the case of $\mu_A=O(1)$, in \cite{BW21} a fine-grained conditional lower bound is established. Roughly speaking, the lower bound states that if $t$ lies out of the region specified by Theorem~\ref{thm:BW21}-$(ii)$, then there does not exist a near-linear time algorithm for the decision problem unless Strong Exponential Time Hypothesis fails. Since this lower bound matches the region specified by Theorem~\ref{thm:BW21}-$(ii)$, it also matches our Theorem~\ref{thm:dense}-$(ii)$. In other words, we solve the search problem of Subset Sum within a tight region that matches the lower bound.

\smallskip
\noindent\textbf{Remark on the accurate values of $c_{\alpha}$, $c_\delta$, $c_\lambda$.} \cite{BW21} sets these constants as  $c_{\alpha}=42480$, $c_\delta=1699200$ and $c_{\lambda}=169920$. Since we are not optimizing these constants, their actual value does not matter much. It is sufficient to know that if $c_{\alpha}$ and $c_{\delta}$ changes by $O(1)$ times, then the argument in \cite{BW21} (which we shall elaborate in the subsequent subsection) still holds by adjusting $c_{\lambda}$ by $O(1)$ times. Since in our Theorem~\ref{thm:BW21}, $C_{\lambda}'$ is by a logarithmic factor larger, this will dominate any $O(1)$ factor. Hence, we will not work with exact values of $c_{\alpha},c_\delta$ and $c_{\lambda}'$, but rather give an estimation whenever necessary. 

Before we proceed, it is useful to observe the following.
By the definition of $\delta$-dense, we have 
\begin{eqnarray}\label{eq:1}
    \mu_Am\le N^2/\delta=\frac{1}{c_\delta}\cdot \frac{N^2}{\log (2N)\log^2(2\mu_A)}.
\end{eqnarray}
Consequently,

\begin{eqnarray}\label{eq:2}
    (4+2C_\lambda')\mu_Am/N^2=O(\frac{1}{\log (2\mu_A)}).
\end{eqnarray}





 The remainder part of this section is devoted to the proof of Theorem~\ref{thm:dense}-$(i)$.
We first discuss the proof of Theorem~\ref{thm:BW21}-$(i)$ in \cite{BW21}, and then we show how to modify that proof for Theorem~\ref{thm:dense}-$(i)$.

\subsection{Proof of Theorem~\ref{thm:BW21}-$(i)$ in \cite{BW21}}
The proof of Theorem~\ref{thm:BW21}-$(i)$ is based on the following three lemmas.

\begin{lemma} [Cf. Theorem 4.1 of \cite{BW21}]\label{lemma:cite-1}
    Given an $\delta$-dense multi-set $A$ of size $N$, in time $\widetilde{O}(N)$ we can compute an integer $\gamma\ge 1$ such that $A'=A(\gamma)/\gamma$ is $\delta$-dense and has no $\alpha$-almost divisor. Moreover, we have the following properties:
    \begin{itemize}
        \item[1.] $\gamma\le 4\mu_A\Sigma_A/|A|^2$;
        \item[2.] $\gamma=O(N)$;
        \item[3.] $|A'|\ge 0.75|A|$;
        \item[4.] $\Sigma_{A'}\ge 0.75\Sigma_A/\gamma$. 
    \end{itemize} 
\end{lemma}

\begin{lemma}[Cf. Theorem 4.2 of \cite{BW21}]\label{lemma:cite-2}
    Let $A$ be a multi-set. If $A$ is $\delta$-dense and has no $\alpha$-almost divisor, then for $\lambda_A:=C_\lambda\cdot  \mu_A m \Sigma_A/|A|^2$, we have
    $$\mathbb{Z}[\lambda_A, \Sigma_A-\lambda_A]\subseteq \mathcal{S}_A.$$
\end{lemma}

\begin{lemma}[Cf. Theorem 4.3 of \cite{BW21}]\label{lemma:cite-3}
    Given a $\delta$-dense multi-set $A$, in time $\widetilde{O}(N)$ we can compute an integer $\gamma\ge 1$ such that for any $t\le \Sigma_A/2$ with $t\ge (4+2C_\lambda)\mu_A m \Sigma_A/|A|^2$:
    
    $t\in \mathcal{S}(A)$ if and only if $(t \mod \gamma) \in (\mathcal{S}_A \mod \gamma)$.
    
\noindent   In particular, $\gamma$ can be chosen as that in Lemma~\ref{lemma:cite-1}. 
\end{lemma}

Recall that our goal is to reconstruct a solution when the answer is "yes".  We remark that \cite{BW21} actually shows how to reconstruct a solution when the answer is "yes", the only problem is that its reconstruction procedure may well exceed linear. In the following, we will describe this reconstruction procedure, and then show how to modify it.

\smallskip
\noindent\textbf{Solution reconstruction procedure in \cite{BW21} when the answer is "yes".}
\smallskip

\noindent Step 1. Apply Lemma~\ref{lemma:cite-1} 
to compute $\gamma$ and obtain $A'=A(\gamma)/\gamma$ in $\widetilde{O}(N)$-time. Observe that $A'$ has no $\alpha$-almost divisor and is $\delta$ dense,  so  Lemma~\ref{lemma:cite-2} is applicable to $A'$. In particular, for any $z\in \mathbb{Z}[\lambda_{A'},\Sigma_{A'}-\lambda_{A'}]$ there exists some $Z\subseteq A(\gamma)$ such that
$$\Sigma_Z=\gamma z.$$

Note that we aim to reconstruct a solution when the answer is "yes". In view of Lemma~\ref{lemma:cite-3}, this means $t \mod \gamma\in \mathcal{S}_A \mod \gamma$.

\smallskip

\noindent Step 2. Find $Y\subseteq \overline{A(\gamma)}$ such that $|Y|\le \gamma$ and $\Sigma_Y \equiv t \quad (\mod \gamma)$. It is shown in \cite{BW21} that such a $Y$ must exist provided that $t \mod \gamma\in \mathcal{S}_A \mod \gamma$. Moreover, $|Y|\le \gamma$ directly implies that $\Sigma_Y\le \gamma m\le 4\mu_Am\Sigma_A/|A|^2$ by Property 1 of Lemma~\ref{lemma:cite-1}.

\smallskip

\noindent Step 3. It is shown in \cite{BW21} that with $Y$ in Step 2, $\gamma| (t-\Sigma_Y)$ and $(t-\Sigma_Y)/\gamma\in \mathbb{Z}[\lambda_{A'},\Sigma_{A'}-\lambda_{A'}]$ both hold. Hence, from Step 1 we know there exists some $Z\subseteq A(\gamma)$ such that $\Sigma_Z=t-\Sigma_Y$, that is, $Y\cup Z$ is the solution.

\subsection{Modifying the solution reconstruction procedure in \cite{BW21}}
Recall that our goal is to  find a solution in $\widetilde{O}(N)$-time for Theorem~\ref{thm:BW21}-$(i)$.
We shall follow the solution reconstruction procedure in \cite{BW21}.

Since Step 1 can be carried out in $\widetilde{O}(N)$-time, it suffices to bound the running time of Step 2 and 3 in $\widetilde{O}(N)$-time. That is, our goal is to find $Y$ and $Z$ in Step 2 and 3 in $\widetilde{O}(N)$-time, respectively.

\subsubsection{Finding $Y$ in $\widetilde{O}(N)$-time.}\label{subsec:Y} We first observe that, under the condition of Theorem~\ref{thm:BW21}-$(i)$, $(4+2C_\lambda)\mu_Am\Sigma_A/N^2\le \Sigma_A/2$. Meanwhile Lemma~\ref{lemma:cite-1} ensures that $\gamma\le 4\mu_A\Sigma_A/N^2$. Simple calculation shows that $\gamma=\widetilde{O}(N)$. So it suffices to find $Y$ in $\widetilde{O}(\gamma)$-time.

Now the problem boils down to the following: among integers in $\overline{A(\gamma)}$, find a subset $Y$ such that $|Y|\le \gamma$ and $\Sigma_Y\equiv t \quad (\mod \gamma)$ in $\widetilde{O}(\gamma)$ time. The proof of Lemma~\ref{lemma:cite-3} in \cite{BW21} already implies an algorithm for it. For the completeness of the paper, we elaborate below. 

$Y$ can be constructed as follows. 
\begin{itemize}
    \item First, we drop the constraint $|Y|\le \gamma$ and simply search for $Y'\subseteq \overline{A(\gamma)}$ with $\Sigma_{Y'}\equiv t \quad (\mod \gamma)$. This is exactly the modular Subset Sum problem. There are several algorithms for this problem that runs in $\widetilde{O}(|\overline{A(\gamma)}|+\gamma)\le \widetilde{O}(N) $-time~\cite{axiotis2019fast, potkepa2021faster, bringmann2021fast, axiotis2021fast}. In particular, the algorithm in~\cite{axiotis2021fast} follows the dynamic programming framework of Bellman~\cite{bellman1966dynamic} that recursively computes the (modular) sumsets $S_i$ that represents all attainable subset sums modulo $d$ for the first $i$ integers, and thus can be used to return a solution in $\widetilde{O}(N)$-time by straightforward backtracking.
    
    \item Next, we shrink the size of $Y'$ in case $|Y'|>\gamma$. This follows the same argument as Claim 4.4 of \cite{BW21}. We order integers of $Y'$ arbitrarily as $y_1,\cdots,y_\ell$, and computes within $\widetilde{O}(\ell)\le \widetilde{O}(N)$ time the partial sums $y_1$, $y_1+y_2$, etc. As long as $\ell>\gamma$, by Pigeonhole we know there must exist some $i<j$ such that $y_1+y_2+\cdots+y_i\equiv y_1+y_2+\cdots+y_j \quad (\mod \gamma)$. Consequently, $y_{i+1}$ to $y_j$ can be removed from $Y'$. Repeatedly applying the above argument we end up with the desired $Y$.
\end{itemize}

\subsubsection{Finding $Z$ in $\widetilde{O}(N)$-time.} This is the challenging part and requires our Theorem~\ref{thm:ss} in this paper.

According to the argument in Step 3, to obtain $Z$ it suffices to establish the following lemma (which is an algorithmic version of Lemma~\ref{lemma:cite-2}).

\begin{lemma}\label{lemma:constructive-cite-2}
    Let $A$ be a multi-set. If $A$ is $\delta$-dense and has no $\alpha$-almost divisor, then for $\lambda_A:=C_\lambda'\cdot  \mu_A m \Sigma_A/|A|^2$ where $C_\lambda':=c_\lambda'\log (2\mu_A)\log (2N)$ for some fixed constant $c_{\lambda}'$, we have
    $$\mathbb{Z}[\lambda_A, \Sigma_A-\lambda_A]\subseteq \mathcal{S}_A.$$
    Moreover, for any $t\in \mathbb{Z}[\lambda_A, \Sigma_A-\lambda_A]$, within $\widetilde{O}(|A|)$-time we can find $B\subseteq A$ such that $\Sigma_B=t$.
\end{lemma}

It is easy to see that if Lemma~\ref{lemma:constructive-cite-2} holds, then Theorem~\ref{thm:dense}-$(i)$ follows since we have found $Y$ in $\widetilde{O}(N)$-time, and Lemma~\ref{lemma:constructive-cite-2} allows us to find $Z$ in  $\widetilde{O}(N)$-time, and $Y\cup Z$ gives the solution.

\subsubsection{Proof of Lemma~\ref{lemma:constructive-cite-2}}
This subsection is devoted to the proof of Lemma~\ref{lemma:constructive-cite-2}. Towards that, we first discuss the proof of Lemma~\ref{lemma:cite-2} in \cite{BW21}, and then we show how it can be adapted for the proof of Lemma~\ref{lemma:constructive-cite-2}.

The key to Lemma~\ref{lemma:cite-2} is the following.

\begin{lemma} [Cf. Theorem 4.35 of \cite{BW21}]\label{lemma:cite-4}
    Let $A$ be $\delta$-dense and has no $\alpha$-almost divisor. Then there exists a partition $A=R\cup P\cup G$ 
    such that:
    \begin{itemize}
        \item the set $\mathcal{S}(P)$ contains an arithmetic progression of $\{s\}+\{0,d,2d,\cdots,2md\}$ where $d\le c_1\cdot \frac{\mu_A\Sigma_A\log (2\mu_A)}{N^2}$ and $s+2md\le c_2\frac{\mu_A m\Sigma_A \log(2\mu_A)}{N^2}$; 
                \item the set $\mathcal{S}(R)$ is $d$-complete, i.e., $\mathcal{S}(R) \pmod d = \mathbb{Z}_d$; 
        \item the multi-set $G$ has sum $\Sigma_G\ge \Sigma_A/2$.
    \end{itemize}
\end{lemma}
Here $c_1$ and $c_2$ are fixed constants. In particular, \cite{BW21} sets $c_1=42480$ and $c_2=84960$. These constants mainly follow from the constant in the finite addition theorems by S{\'a}rk{\"o}zy. The actual value of $c_1$ and $c_2$ affects the values $c_\alpha$ and $c_\delta$ before. But as we mentioned, when these values change by $O(1)$ times, Theorem~\ref{thm:BW21} still holds by suitably adjusting $c_{\lambda}$ by $O(1)$ times, and an $O(1)$ factor is dominated by that our $C_{\lambda}'$ blows up $C_{\lambda}$ by a logarithmic factor.

Given Lemma~\ref{lemma:cite-4}, we sketch the proof strategy for Lemma~\ref{lemma:cite-2} and details will follow: we first pick arbitrary elements in $G$ that sum to some $t'=t-\Theta(\lambda_A)$. It remains to hit $t-t'=\Theta(\lambda_A)$ using $P\cup R$. We pick elements from $R$ that sum to $t-t'-s$ modulo $d$. It remains to add some $s+kd$ for some suitable $k$, and this can be achieved as $\mathcal{S}(P)$ contains a sufficiently long AP. Indeed, the proof of Theorem 4.2 in \cite{BW21} showed that within $\widetilde{O}(N)$-time we can compute $(p,r,g)\in P\times R\times G$ such that $t=p+r+g$. The only bottleneck is to find a witness for $p$ efficiently, which we handle by our Theorem~\ref{thm:ss}.

 In the following we show that, $P$, $R$ and $G$ can be efficiently constructed together with a witness, thus concluding Lemma~\ref{lemma:constructive-cite-2}. 
 
 We start with the remainder set $R$. 

\paragraph{Remainder set $R$ - definition and existence.}
\begin{lemma}[Cf. Theorem 4.20 of \cite{BW21}] \label{lemma:cite-R-1}
    Let $A$ be a $\delta$-dense multi-set of size $N$ within $[1,m]$ that has no $\alpha$-almost divisor. Then there exists a subset $R\subseteq A$ such that  
    \begin{itemize}
        \item $|R|\le |A|\cdot 8\alpha/\delta \cdot \log(2N)$;
        \item $\Sigma_R\le \Sigma_A\cdot 8\alpha/\delta\cdot \log (2N)$;
        \item for any integer $1<\gamma\le \alpha \mu_A\Sigma_A/N^2$, the multi-set $R$ contains at least $\gamma$ integers not divisible by $\gamma$, that is, $|\overline{R(\gamma)}|\ge \gamma$.
    \end{itemize}
\end{lemma}
The existence of such an $R$ is already shown in \cite{BW21}. Indeed, \cite{BW21} explicitly constructs $R$ as follows.

\paragraph{Construction of $R$.} Set $\tau:=\lceil \alpha\cdot \mu_A\Sigma_A/N^2\rceil$.  
\begin{itemize}
    \item Pick an arbitrary subset $R'\subseteq A$ of size $2\tau$. 
    \begin{itemize}
        \item $R'$ can be  constructed trivially in $\widetilde{O}(N)$-time.
    \end{itemize}
    \item Let $P$ be the set of primes $p$ with $p\le \tau$ and $|\overline{R'(p)}|<\tau$. \cite{BW21} showed (in Claim 4.21) that $|P|\le 2\log m$. 
    \begin{itemize}
        \item \cite{BW21} showed in Theorem 3.8 that the prime factorization of $N$ integers within $m$ can be computed in time $\widetilde{O}(N+\sqrt{m})$. 
        By Eq~\eqref{eq:1} we know $\sqrt{m}=\widetilde{O}(N)$. That is, within $\widetilde{O}(N)$-time we can obtain the prime factorization of all integers in $A$. With the prime factorization, $|R'(p)|$'s can be computed as follows. List integers in $R'$ arbitrarily. Among the first $i$ integers of $R'$, let $\xi_{i,p}$ be the number of multiples of $p$. Given $\xi_{i,p}$'s, computing $\xi_{i+1,p}$'s requires $O(\log m)$-time since every integer has at most $O(\log m)$ prime factors. Overall, $|R'(p)|$'s, and hence $|\overline{R'(p)}|$'s, and hence $P$, can be computed in $\widetilde{O}(N)$-time.
    \end{itemize}
    \item For any $p\in P$, let $R_p\subseteq \overline{A(p)}$ be an arbitrary set of size $\tau$. Then $R:=R'\cup \bigcup_{p\in P} R_p$.
    \begin{itemize}
        \item $R_p$ can be constructed trivially in $O(\tau)$-time. Hence, $R$ is constructed in $\widetilde{O}(N)$-time. 
    \end{itemize} 
\end{itemize}

 \paragraph{Witness for $R$.}
\cite{BW21} showed in Theorem 4.22 that for every integer $1<\gamma\le \alpha \mu_A\Sigma_A/N^2$, and every integer $0\le \eta<\gamma$, there exists some $R^*\subset R$ such that $\Sigma_{R^*}\equiv \eta \pmod \gamma$. We further show that such an $R^*$ can be found in $\widetilde{O}(N)$-time. Indeed, finding $R^*$ can be done in exactly the same way as finding $Y$ in Subsection~\ref{subsec:Y}.

Next, we discuss the AP set $P$.

\paragraph{Construction of $P$.} We first find $R\subseteq A$ in $\widetilde{O}(N)$-time. Using that $|R|\le |A|\cdot 8\alpha/\delta \cdot \log(2N)$ from Lemma~\ref{lemma:cite-R-1}, we have $|R|\le N/4$. 

We actually need a slightly stronger version of our Theorem~\ref{thm:ss}, which is on multi-sets instead of sets. 


\begin{lemma}\label{lem:multiset}
  Let $B$ be a multi-set set of $N$ integers from $\mathbb{Z}[1, m]$. For any integer $\ell$ with 
\begin{eqnarray}\label{eq:a}
        m \leq \ell \leq \Theta(\frac{N^2}{\mu_B\log^2 \mu_B\log N}),
\end{eqnarray}
    in $O(N \log N)$ time, we can compute a subset $B' \subseteq B$ with $|B'| \leq O(\frac{\ell\mu_B\log \mu_B\log N}{N})$, an arithmetic progression
    \[
    \{s\} + \{0, d, 2d, \ldots, \ell d\} \in \mathcal{S}(B')
    \]
    with $d \leq O(\frac{m\mu_B\log \mu_B}{N})$, and a witness with $O(\frac{\ell \mu_B\log \mu_B\log N}{N})$ expected query time for this arithmetic progression.
\end{lemma}   


\begin{proof}[Proof sketch] Lemma~\ref{lem:multiset} can be proved by following the guideline below.
\begin{itemize}
	\item[(i).] At the cost of losing factor $\log \mu_B$ in cardinality, we may assume every distinct integer in $B$ has the same multiplicity $\mu$. This can be achieved as follows: 
	Based on the multiplicity, we divide $B$ into disjoint subsets $B_0, B_1,B_2,\cdots$, where $B_i$ consists all distinct elements in $B$ whose multiplicity lies in $[2^{i},2^{i+1}]$. Pick $i$ such that $2^i|B_i|$ is the maximal and let it be $i^*$. Reducing the multiplicity of every distinct element in $B_{i^*}$ to exactly $2^{i^*}$. Now note that  $2^{i^*}|B_{i^*}|=\Omega(|B|/\log \mu_B) $.
	\item[(ii).] Generating a short arithmetic progression:
	Replacing $A$ with $B$ in Subsection~\ref{sec:sumset2ss}, we obtain an analogy of Lemma~\ref{lem:short-ap-ss} which computes $B^*\subseteq B$ such that $S(B^*)$ admits an arithmetic progression of length $\ell$ for 
	$$\Theta(\frac{m}{|B|})\le \ell \le \Theta(\frac{|B|}{\log|B|}).$$ 
	\item[(iii).] Augmenting a short arithmetic progression: replacing $A$ with $B$ in Subsection~\ref{sec:aug-ss} while noticing that every distinct integer has $\mu$ copies, we obtain 
	an analogy of Lemma~\ref{lem:aug-one-step} stating that we can compute $B'\subseteq B$ with $|B'|=O(\frac{d}{\mu}+\frac{\gamma}{\mu})$ that can be used to increase the length of a progression $P$ from $\ell=\Omega(\frac{m}{|B|})$ to at least $\frac{3}{2}\ell$.
	\item[(iv).] Following Lemma~\ref{lem:long-ap-ss}, we split $B$ into two parts, apply our argument in (ii) to the first part to generate a short arithmetic progression, and apply our argument in (iii) to iteratively augment the arithmetic progression. 
\end{itemize}
Since the whole proof essentially repeats Section~\ref{sec:ss}, we omit it.
\end{proof}

Apply Lemma~\ref{lem:multiset} by setting $B$ as the smallest $N/4$ elements from $A\setminus R$ and $\ell =2m$. We can compute $P\subseteq A\setminus R $ such that $|P|\le \Theta(\frac{m\mu_A\log \mu_A\log N}{N})\le N/4$, together with some $s$ and $d\le O(\frac{m\mu_A\log \mu_A}{N})$ such that
 \[
 \{s\} + \{0, d, 2d, \ldots, 2m d\} \in \mathcal{S}(P)
 \]
 
 \begin{observation}
    The way we set $B$ implies that $\Sigma_P\le \Sigma_{A\setminus R}/3$.
 \end{observation}
 
 \noindent\textbf{Remark on the applicability of Lemma~\ref{lem:multiset}.} Lemma~\ref{lem:multiset} is applicable when Eq~\eqref{eq:a} holds. So setting $\ell=2m$ and $B$ as the smallest $N/4$ elements of $A\setminus R$  requires as a precondition that $2m\le \Theta(\frac{N^2}{\mu_A\log^2 \mu_A\log N})$. Given Eq~\eqref{eq:1}, this can be fulfilled by adjusting the constant $c_\delta$.  
 
 \noindent\textbf{Remark on upper bounding $s+2md$.} For the sake of the subsequent proof, we need to upper bound $s+2md$. Denote by $p$ the largest element in $P$, we get that $s+2md\le p|P|$. 
We know that $|P|\le \Theta(\frac{m\mu_A\log \mu_A\log N}{N})$. It remains to bound $p$. Since $P\subseteq B$ and $B$ consists of the smallest $N/4$ elements of $A\setminus R$, we have that 
$$p\le \frac{\Sigma_{A\setminus R}}{|A\setminus R|}\le \frac{\Sigma_{A}}{3N/4}\le 2\frac{\Sigma_A}{N}.$$
Consequently,
\begin{eqnarray}\label{eq:b}
    s+2md\le p|P|\le \Theta(\frac{m\mu_A\Sigma_A\log \mu_A\log N}{N^2}).
\end{eqnarray}

Let us compare the bound given by Eq~\eqref{eq:b} with that given in \cite{BW21}. A crucial fact is that $s+2md$ can be bounded by $O(\frac{m\mu_A\Sigma_A\log \mu_A}{N^2})$ in \cite{BW21} (see Claim 4.37 of \cite{BW21}). We see that, this differs by a factor of $\log N$ with Eq~\eqref{eq:b}. This is the main reason that our Theorem~\ref{thm:BW21} needs to set $C_\lambda'$ by a logarithmic factor larger than $C_{\lambda}$.  The reason that Eq~\eqref{eq:b} is not as sharp as that in \cite{BW21} lies in the size of $B'$ in Lemma~\ref{lem:multiset}.  Eq~\eqref{eq:b} would have been as sharp if we can shed off $\log N$ there, but this seems not possible when we need an efficient witness.

\paragraph{Construction of $G$.} Set $G:=A\setminus (P\cup R)$. Given that $\Sigma_R\le \Sigma_A\cdot 8\alpha/\delta\cdot \log (2N)$ and that $\Sigma_P\le \Sigma_{A\setminus R}/3$, one can verify that $\Sigma_G\ge \Sigma_A/2$ (see Claim 4.38 of \cite{BW21}).

\paragraph{Finalizing the proof of Lemma~\ref{lemma:constructive-cite-2}}
Let $t\in \mathbb{Z}[\lambda_A, \Sigma_A-\lambda_A]$ where $\lambda_A:=C_\lambda' \cdot \mu_Am\Sigma_A/N^2$, and $C_\lambda'=c_\lambda' \log (2\mu_A)\log (2N)$ for some constant $c_{\lambda}$. By symmetry we may assume $t\le \Sigma_A/2$. 

We shall maintain a bucket and add elements of $A$ into it until the sum of these elements becomes $t$. We first pick $G'\subseteq G$ by greedily adding elements of $G$ into the bucket until 
$$t-s-d(m+1)-m<\Sigma_{G'}\le t-s-d(m+1).$$
This is possible because: (i). any element is at most $m$ and $\Sigma_{G}\ge \Sigma_A/2$, and (ii). $s+md\le \Theta(\frac{m\mu_A\Sigma_A\log \mu_A\log N}{N^2})\le \lambda_A$ by taking a suitable constant $c_\lambda'$. 

Next we pick a subset $R'\subseteq R$ that sums to $(t-\Sigma_{G'}-s)$ modulo $d$. This is possible because $R$ is $d$-complete, and we can further restrict that $|R'|\le d$ (using the same argument as we argue on $Y'$ in Subsection~\ref{subsec:Y}). Hence, $\Sigma_{R'}\le md$. We thus have $t-\Sigma_{G'\cup R'} \equiv s \quad (\mod d)$, and 
$$s+d\le t-\Sigma_{G'\cup R'}\le s+d(m+1)+m.$$
It is obvious that we can select find $P'\subseteq P$ with $\Sigma_{P'}=t-\Sigma_{G'\cup R'}$. Moreover, such a $P'$ can be found in $\widetilde{O}(N)$-time given the witness in Lemma~\ref{lem:multiset}.

To summarize,  Lemma~\ref{lemma:constructive-cite-2} is true, and consequently Theorem~\ref{thm:BW21} is true.

\bibliographystyle{alphaurl}
\bibliography{main}
\end{document}